\documentclass[journal]{IEEEtran}
\IEEEoverridecommandlockouts

\usepackage{amsmath}
\usepackage[linesnumbered, ruled]{algorithm2e}
\usepackage{amsfonts}
\usepackage{amsthm}
\usepackage{algorithmicx}
\usepackage{array}
\usepackage[caption=false,font=normalsize,labelfont=sf,textfont=sf]{subfig}
\usepackage{textcomp}
\usepackage{stfloats}
\usepackage{url}
\usepackage{verbatim}
\usepackage{graphicx}  
\usepackage{float}  
\usepackage{subfloat}  

\usepackage{cite}
\usepackage{multicol}

\usepackage{bm}
\usepackage{booktabs}
\usepackage{amssymb}
\usepackage{comment}
\usepackage[colorlinks=true, allcolors=blue]{hyperref}
\usepackage{mathrsfs}
\usepackage{tikz}
\usepackage{pifont}
\newtheorem{lemma}{Lemma}
\newtheorem{theorem}{Theorem}
\begin{document}

\setlength{\textfloatsep}{0\baselineskip plus 0.2\baselineskip minus 0.2\baselineskip}


\title{On the Distribution of SINR for Cell-Free Massive MIMO Systems
	{\footnotesize \textsuperscript{}}
}

\author{Baolin Chong, Fengqian Guo, Hancheng Lu,~\IEEEmembership{Senior Member,~IEEE} and Langtian Qin
\IEEEcompsocitemizethanks{\IEEEcompsocthanksitem Baolin Chong, Fengqian Guo, Hancheng Lu and Langtian Qin are with the Department of Electronic Engineering and Information Science, University of Science and Technology of China, Hefei 230027, China. (e-mail: chongbaolin@mail.ustc.edu.cn; fqguo@ustc.edu.cn; hclu@ustc.edu.cn; qlt315@mail@mail.ustc.edu.cn) \protect\\}}


\maketitle

\begin{abstract}

Cell-free (CF) massive multiple-input multiple-output (mMIMO) has been considered as a potential technology for Beyond 5G communication systems. However, the performance of CF mMIMO systems has not been well studied. Most existing analytical work on CF mMIMO systems is based on the expected signal-to-interference-plus-noise ratio (SINR). The statistical characteristics of the SINR, which is critical for emerging applications that focus on extreme events, have not been investigated. To address this issue, in this paper, we attempt to obtain the distribution of SINR in CF mMIMO systems. Considering a downlink CF mMIMO system with pilot contamination, we first give the closed-form expression of the SINR. Based on our analytical work on the two components of the SINR, i.e., desired signal and interference-plus-noise, we then derive the probability density function and cumulative distribution function of the SINR under maximum ratio transmission (MRT) and full-pilot zero-forcing (FZF) precoding, respectively. Subsequently, the closed-form expressions for two more sophisticated performance metrics, i.e., achievable rate and outage probability, can be obtained. Finally, we perform Monte Carlo simulations to validate our analytical work. The results demonstrate the effectiveness of the derived SINR distribution, achievable rate, and outage probability.

\end{abstract}


\section{Introduction}
Multiple-input multiple-output (MIMO) technology has been extensively deployed in current communication systems to support diverse service types \cite{SurveyZheng2015,pei2022key,FengRecent2023}.
However, the arrival of Beyond 5G has introduced a multitude of emerging applications, such as holographic telepresence, augmented reality, virtual reality, and the Internet of everything \cite{CellElhoushy2022}, which have imposed more stringent demands on system performance.
The traditional cell-centric MIMO-based cellular networks have become inadequate to meet the demands of the current advancements.
In traditional cellular networks, where users within a cell are served by a single base station (BS), users at the cell edge often experience considerable inter-cell interference caused by signals from neighboring cells. Hence, inter-cell interference is an important factor limiting the capacity of cellular networks \cite{FundamentalLozano2013}.
Furthermore, users located at different positions within the cellular network face challenges in receiving uniform service due to varying distances from the BS.

To overcome aforementioned issues, cell-free (CF) massive MIMO (mMIMO) has been proposed, which combines the advantages of mMIMO and distributed MIMO \cite{GuoJoint2022}. In a CF mMIMO system, a central processing unit (CPU) is employed to control multiple access points (APs), enabling the concurrent service of multiple users on the same time-frequency resources \cite{ZhangCell2019,CellElhoushy2022}.
Particularly, the distributed deployment of multiple APs eliminates the cell boundaries and reduces the distance between APs and users \cite{ChenChannel2018}, with which inconsistent service quality and large-scale fading effects can be effectively avoided. It has been proved both spectral efficiency (SE) \cite{NayebiPrecoding2017} and energy efficiency (EE) \cite{NgoOn2018} are substantially improved in the CF mMIMO systems.

The significant technical advantages demonstrated by CF mMIMO have spurred an increased research focus on its performance analysis. For CF mMIMO systems under maximum ratio transmission (MRT) precoding, the derivation of expected signal-to-interference-plus-noise ratio (SINR) was initially undertaken by authors of \cite{NgoCell2017}.
Subsequently, in \cite{InterdonatoLocal2020} and \cite{ZhangLocal2021}, the CF mMIMO system's performance is investigated under downlink and uplink data transmission scenarios with different zero-forcing (ZF) precoding and combining schemes.
The performance has also been studied for CF mMIMO systems with other linear receivers, e.g., minimum mean square error (MMSE) receiver \cite{NayebiPerformance2016}.
In \cite{WangUplink2022}, the analytical work was extended to encompass multi-antenna users, coherent Rayleigh fading channels, and achievable rate, building upon previous considerations of single-antenna users and Rayleigh fading channels.
The impact of asynchronous reception of signals by users, resulting from the distributed deployment of APs, was addressed in \cite{ZhengAsynchronous2023}, where a closed-form expression for SE was derived.
Another aspect of performance analysis focused on the practicality of CF mMIMO systems. In \cite{PapazafeiropoulosPerformance2020}, a stochastic geometry approach was utilized to model the locations of APs in CF mMIMO systems, and an expression for achievable rates was derived.
Furthermore, authors in \cite{zdoganPerformance2019} and \cite{FangCell2021} investigated the impact of mobility and oscillator phase noise on system performance, respectively. In \cite{MasoumiPerformance2020}, authors considered limited feedback capacity as well as hardware impairments in both APs and users, which leads to the derivation of the corresponding expected SINR.
Moreover, the performance of CF mMIMO with other technologies, including low-resolution analog-to-digital converter (ADC) \cite{ZhangOn2019}, reconfigurable intelligent surfaces (RIS) \cite{AlRIS2021}, simultaneous wireless information and power transfer (SWIPT) \cite{KusaladharmaStochastic2020}, non-orthogonal multiple access (NOMA) \cite{LiNOMA2018}, and rate-splitting multiple-access (RSMA) \cite{MishraRate2022}, has also been extensively studied.

Existing analytical work provides a comprehensive and detailed characterization of the average performance of CF mMIMO systems. However, capturing the statistical properties of the SINR, which is a critical task in CF mMIMO systems applied for emerging applications that focus on extreme events, has not been addressed.
The exploration of SINR distribution has been a prominent research topic in MIMO systems.
The SINR distribution for MMSE MIMO systems was derived by authors of \cite{PingOn2006}, and subsequent advancements by \cite{LimOn2019} resulted in the derivation of the exact SINR distribution.
Based on the analyzed SINR distribution, numerous research has delved into the performance evaluation of MIMO systems, encompassing metrics such as SE and EE \cite{NgoEnergy2013}, uncoded error and outage probabilities, diversity-multiplexing gain tradeoff, and coding gain \cite{JiangPerformance2011}.
As MIMO systems evolve into CF mMIMO systems, the foundational performance analysis concerning SINR distribution continues to be deficient.
Besides, based on the existing research on the expected SINR, only the lower performance bounds of the CF mMIMO system can be obtained, such as the lower bound of achievable rates. However, this fails to exactly reflect the actual system performance. To more precisely characterize the system's performance, it becomes imperative to acquire the SINR distribution.
Furthermore, for emerging applications focused on ultra-reliable low-latency communication, such as intelligent driving, remote healthcare, and tactile internet, analyzing the system's performance under extreme conditions becomes even more crucial, clearly going beyond what can be achieved by relying solely on the expected SINR. Analyzing the statistical performance of the CF mMIMO system becomes paramount in such cases.
Analyzing the distribution of SINR is becoming increasingly crucial.
While there are already existing studies on the distribution of signal-to-noise ratio (SNR) in CF mMIMO systems \cite{ZhangPerformance2023,ZhangStatistical2021}, the impact of interference remains challenging to disregard in many cases.
Therefore, conducting research on the statistical properties of SINR in CF mMIMO systems is necessary to gain a deeper understanding of their characteristics.

In this paper, we investigate the statistical characteristics of the SINR in CF mMIMO systems. It should be noted that such analytical work is non-trivial. The dense deployment of APs introduces complexities due to the increased number of signals received by all users. These signals include both target signals and interference signals originating from multiple APs.
Each received signal represents the accumulation of numerous individual signals, resulting in intricate interactions among different signals. Furthermore, the reuse of pilot sequences among different users in the system leads to channel estimation correlation between them. The correlation introduces coherence in the precoding vectors and causes the signals transmitted to users employing the same pilot to exhibit coherence. We overcome the aforementioned challenges and obtain the distribution of the SINR in CF mMIMO systems. The main contributions are summarized as follows:

\begin{itemize}
	
	\item We consider data transmission in a downlink CF mMIMO system with pilot contamination. Based on the received signals at the users, we give the expressions of the SINR when MRT and full-pilot ZF (FZF) precoding schemes are involved.	
	
	\item We conduct an analysis of the SINR distribution under MRT and FZF precoding. Specifically, we begin by decomposing the SINR into two components: the desired signal (DS) and the interference plus noise (IN).
For MRT precoding, we leverage the Central Limit Theorem (CLT) and random matrix theory to separately analyze the distributions of DS and IN. This allows us to derive the PDF and CDF of SINR, taking into account the independence between DS and IN.
In the case of FZF precoding, we directly compute the distribution of DS and analyze the distribution of IN accordingly. Subsequently, we derive the overall distribution of SINR based on these individual components.

	\item In order to compare with the lower bound of the achievable rate, we derive closed-form expressions for the achievable rate and the outage probability based on our analytical work on the SINR distribution under the MRT and FZF precoding, respectively. To validate our derived results, we conduct Monte Carlo simulations. The simulation results demonstrate the effectiveness of our analytical work on the distribution of SINR, achievable rate, and outage probability.
\end{itemize}

The rest of the paper is organized as follows. Section II gives the downlink CF mMIMO system with pilot contamination and the expression of SINR. The distribution of SINR under MRT and FZF precoding for the CF mMIMO system is obtained in Section III, respectively.
Section IV gives closed-form expressions for the achievable rate and the outage probability under the MRT and FZF precoding, respectively.
Section V presents the simulation results and analysis. Finally, the conclusion is drawn in Section VI.

\emph{Notations:} In this paper, vectors, and matrices are denoted by lowercase and uppercase bold letters, respectively.
For a general matrix $\mathbf{A}$, $\mathbf{A}^H$ and $\mathbf{A}^{-1}$represent the Hermitian and inverse of $\mathbf{A}$, respectively.
$\left | \cdot \right | $ and $\left ( \cdot \right )^* $ denote the modulus and the conjugate of a complex number, respectively.
$\left \| \mathbf{x}  \right \|$ denotes $\ell_2$ norm of vector $\mathbf{x}$ and $\mathbb{C}^{x\times y}$ represents the space of $x\times y$ complex number matrices.
$\mathbb{E}\left \{ x \right \}$ denote the expected value of $x$ and the calligraphy upper-case letter, such as $\mathcal{K}$, denotes a set.

\section{System Model}\label{Section2}

\begin{table}
	\centering
	\caption{List of Key Notations}
	\begin{tabular} {m{32pt}	m{180pt}	m{0cm}}
		\toprule
		Symbol		& Description							 	\\
		\midrule

		$\mathcal{M}$, $\mathcal{K}$	     & Index set of APs and users								\\
        $N$				                     & Number of antennas each AP        \\		
        $\mathbf{h}_{mk}$ 	                 & Actual channel between AP $m$ and user $k$         \\
        $\hat{\mathbf{h}}_{mk}$              & Estimated channel between AP $m$ and user $k$ \\
        $\bar{\mathbf{h}}_{mk}$              & Channel estimated error between AP $m$ and user $k$ \\
        $\bar{\mathbf{H}}_{m}$               & Full-rank matrix of channel estimates at AP $m$\\
        $\beta_{mk}$                         & Large scale fading between AP $m$ and user $k$ \\
        $c_{mk}$                             & Estimated large scale fading between AP $m$ and user $k$ \\
        $l_p$                                & Length of pilot sequence\\
        $\mathcal{P}_k$                      & Index set of users which use the same pilot sequence with user $k$\\
        $\rho_p$, $\rho_d$                   & Normalized pilot power and normalized downlink transmission power each AP\\
        $\mathbf{b}_{mk}$                    & Precoding matrix at AP $m$ for user $k$\\
        $\text{DS}_k$, $\text{IN}_k$         & $\text{DS}$ and $\text{IN}$ for user $k$\\
		\bottomrule
	\end{tabular}
	\label{t1}
\end{table}

We consider a downlink CF mMIMO system operating in time-division duplex (TDD), where a total of $M$ APs, each equipped with $N$ antennas, have the capability to jointly and coherently serve $K$ users, each equipped with a single antenna.
The set of APs and users are denoted by $\mathcal{M}$ and $\mathcal{K}$, respectively. All APs are connected to a CPU through backhauls. The channel vector $\mathbf{h}_{mk}$ between AP $m$ and user $k$ is modeled as
\begin{equation}\label{channel model}
  \mathbf{h}_{mk} = \beta_{mk}^{\frac{1}{2} }\mathbf{g}_{mk},
\end{equation}
where $\beta_{mk}$ represents the large-scale fading coefficient which is influenced by path loss and shadowing fading, and $\mathbf{g}_{mk} \sim \mathcal{CN}(0,\mathbf{I}_N)$ denotes the small-scale fading.
Table \ref{t1} lists a compilation of the main notations used in this paper.

\subsection{Uplink Training and Channel Estimation}
During the uplink training phase, each user simultaneously transmits a pilot sequence to all APs, with the length of the pilot sequence denoted as $l_p$.
Define $i_k\in\left \{ 1,\dots,l_p \right \}$ as the index of the pilot used by user $k$, represented by $\pmb{\phi}_{i_k} \in \mathbb{C}^{l_p \times 1}$.
We assume that $l_p \le K$, which indicates that some users will share the same pilot sequence. Consequently, we define $\mathcal{P}_k \subset \mathcal{K}$ as the set of indices, including $k$, corresponding to users assigned the same pilot as user $k$. Therefore, for any user $k_1$ where $k_1 \ne k$, the condition $i_k =i_{k_1} \Leftrightarrow k_1 \in \mathcal{P}_k$ holds true.
The pilot sequences are mutually orthogonal and the inner product of pilot sequences is given by
\begin{equation}\label{pilot sequence expression}
  \pmb{\phi}_{i_{k_1}}^{H}\pmb{\phi}_{i_k} = \left \{
  \begin{aligned}
   &0,   && k_1 \notin \mathcal{P}_k,\\
   &l_p, && k_1 \in \mathcal{P}_k.
  \end{aligned}
    \right .
\end{equation}
The pilot signal received by AP $m$, denoted as $\mathbf{Y}_m^p$, can be expressed as
\begin{equation}\label{pilot signal received by AP}
  \mathbf{Y}_m^p = \sqrt{\rho_p}\sum_{k\in \mathcal{K}}\mathbf{h}_{mk}\pmb{\phi}_{i_k}^H + \mathbf{Z}_m^p,
\end{equation}
where $\rho_p$ denotes the normalized pilot power, and $\mathbf{Z}_m^p \in \mathbb{C}^{N\times l_p}$ corresponds to the additive Gaussian noise matrix at AP $m$ during the uplink training phase, each element of which is independent and follows the distribution of $\mathcal{CN}(0,1)$.

To estimate the channel to user $k$, AP $m$ performs a correlation operation between the received pilot signal and the corresponding pilot sequence $\pmb{\phi}_{i_k}$. Subsequently, it applies the MMSE technique to obtain the MMSE channel estimate $\hat{\mathbf{h}}_{mk}$ \cite{InterdonatoLocal2020}, which can be calculated as
\begin{equation}\label{MMSE estimated channel}
  \hat{\mathbf{h}}_{mk} = \kappa_{mk}\mathbf{Y}_m^p\pmb{\phi}_{i_k},
\end{equation}
where $\kappa_{mk}$ is defined as
\begin{equation}\label{define of kappa}
  \kappa_{mk} = \frac{\sqrt{\rho_p}\beta_{mk}}{l_p\rho_p{\textstyle \sum_{k_1 \in \mathcal{P}_k}\beta_{mk_1}}+1 }.
\end{equation}
The channel estimation follows the distribution of $\mathcal{CN}(0, c_{mk}\mathbf{I}_N)$ with $c_{mk}$ given by \cite{InterdonatoLocal2020}
\begin{equation}\label{estimated coefficient}
  c_{mk} = \frac{l_p\rho_p\beta_{mk}^2}{l_p\rho_p{\textstyle \sum_{k_1 \in \mathcal{P}_k}\beta_{mk_1}}+1 }.
\end{equation}
Then, denote $\bar{\mathbf{h}}_{mk} = \mathbf{h}_{mk} - \hat{\mathbf{h}}_{mk}$ as the channel estimation error, which is independent of $\hat{\mathbf{h}}_{mk}$ and follows the distribution of $\mathcal{CN}(0, (\beta_{mk}-c_{mk})\mathbf{I}_N)$.
When multiple users share the same pilot sequence, their channel estimates are parallel to each other. This relationship can be expressed as
\begin{equation}\label{parallel channel}
  \hat{\mathbf{h}}_{mk} = \frac{\beta_{mk}}{\beta_{mk_1}} \hat{\mathbf{h}}_{mk_1},\ k_1 \in \mathcal{P}_k.
\end{equation}

\subsection{Downlink Data Transmission}
During downlink data training, all APs transmit signals to all users. The signal transmitted by AP $m$ is given by
\begin{equation}\label{downlink signal from AP}
  \mathbf{y}_m^d = \sqrt{\rho_d}\sum_{k\in \mathcal{K}}\sqrt{\eta_{mk}}\mathbf{b}_{mk}q_k,
\end{equation}
where $\rho_d$ represents the normalized downlink transmission power each AP and $\eta_{mk}$ denotes power control coefficient from AP $m$ to user $k$.
The symbol $q_k$, which satisfies $\mathbb{E}\left \{ \left | q_k \right |^2  \right \} = 1$, is the symbol intended for the user $k$ during downlink data transmission and $\mathbf{b}_{mk}$ is precoding vector.

The precoding vector $\mathbf{b}_{mk}$ for user $k$ at AP $m$ is determined based on the channel estimation.
The matrix of the channel estimates, $\hat{\mathbf{H}}_{m} = \left [ \hat{\mathbf{h}}_{m1},\hat{\mathbf{h}}_{m2},\cdots,\hat{\mathbf{h}}_{mK} \right ] \in \mathbb{C}^{N\times K} $, is rank-deficient due to the pilot sequence length being smaller than the number of users.
To obtain a full-rank matrix of channel estimates, denoted as $\bar{\mathbf{H}}_{m} \in \mathbb{C}^{N\times l_p}$, we define $\bar{\mathbf{H}}_{m} = \mathbf{Y}_m^p \pmb{\Phi}$, where $\pmb{\Phi} = \left [ \pmb{\phi}_1,\dots,\pmb{\phi}_{l_p} \right ] \in \mathbb{C}^{l_p\times l_p} $ is the pilot-book matrix \cite{InterdonatoLocal2020}. The channel estimate between AP $m$ and user $k$ can be expressed in terms of $\bar{\mathbf{H}}_{m}$ as
\begin{equation}\label{channel estimated expression 2}
  \hat{\mathbf{h}}_{mk} = \kappa_{mk}\bar{\mathbf{H}}_{m}\mathbf{\mathbf{e}_{i_k}},
\end{equation}
where $\mathbf{e}_k$ represent the $i_k$th column of unit matrix $\mathbf{I}_{l_p}$. We consider using MRT and FZF precoding during downlink data transmission \cite{NgoCell2017,InterdonatoLocal2020}, the precoding vector $\mathbf{b}_{mk}$ for user $k$ can be expressed as follows:
\begin{equation}\label{proceding vector}
  \mathbf{b}_{mk} = \left \{
  \begin{aligned}
    &\frac{\bar{\mathbf{H}}_{m}\mathbf{e}_{i_k}}{\sqrt{\mathbb{E} \left \{ \left \| \bar{\mathbf{H}}_{m}\mathbf{e}_{i_k} \right \|^2 \right \}}   }, && \text{MRT},\\
    &\frac{\bar{\mathbf{H}}_{m}\left ( \bar{\mathbf{H}}_{m}^H\bar{\mathbf{H}}_{m} \right )^{-1} \mathbf{e}_{i_k}}{\sqrt{\mathbb{E} \left \{ \left \| \bar{\mathbf{H}}_{m}\left [ \bar{\mathbf{H}}_{m}^H\bar{\mathbf{H}}_{m} \right ]^{-1} \mathbf{e}_{i_k}\right \|^2 \right \}}   }, && \text{FZF}.\\
  \end{aligned}
   \right .
\end{equation}
Each user receives the signals from all APs and the received signal at user $k$ is
\begin{equation}\label{downlink signal to UE}
  r_k^d = \sqrt{\rho_d}\sum_{m\in\mathcal{M} } \sum_{k_1\in \mathcal{K}} \sqrt{\eta_{mk_1}} \mathbf{h}_{mk}^H\mathbf{b}_{mk_1}q_{k_1} + z_k,
\end{equation}
where $z_k$ is the noise at user $k$ with the distribution of $\mathcal{CN}(0,1)$.
Then, the SINR of user $k$ during downlink data transmission is given in (\ref{uplink SINR exprssion}).
\begin{figure*}
    \begin{equation}\label{uplink SINR exprssion}
      \begin{aligned}
         \gamma_k = \frac{\rho_d\left |\sum_{m\in \mathcal{M}}\sqrt{\eta_{mk}} \hat{\mathbf{h}}_{mk}^H\mathbf{b}_{mk}\right |^2 }{\rho_d\sum_{k_1\ne k}\left |\sum_{m\in \mathcal{M}}\sqrt{\eta_{mk_1}}  \hat{\mathbf{h}}_{mk}^H\mathbf{b}_{mk_1}\right |^2 + \rho_d\sum_{k_1\in \mathcal{K}}\left |\sum_{m\in \mathcal{M}}\sqrt{\eta_{mk_1}}  \bar{\mathbf{h}}_{mk}^H\mathbf{b}_{mk_1}\right |^2 + z_k} .
      \end{aligned}
    \end{equation}
	{\noindent} \rule[-10pt]{18cm}{0.05em}
\end{figure*}

\section{Distribution of SINR for CF mMIMO System}
In this section, we derive the distribution of the SINR for the CF mMIMO system. To begin, we partition the received signal from the user into two components: the DS and the IN, utilizing the expression of SINR. Subsequently, we investigate the distribution of DS and IN under the MRT precoding scheme, employing the CLT and random matrix theory. By leveraging the independence properties of DS and IN, we obtain the distribution of SINR.
Furthermore, we proceed to change the precoding vector, transitioning to the FZF scheme. We calculate the value of DS and obtain the distribution of IN. Finally, based on the aforementioned analysis
regarding DS and IN, we derive the distribution of SINR for the CF mMIMO system.

Based on the expression of SINR for user $k$, we divide the signal that user $k$ received into DS and IN. Then the expression of SINR $\gamma_k$ can be represented as
\begin{equation}\label{represent of gamma}
  \gamma_k = \frac{\text{DS}_k}{\text{IN}_k} = \frac{\rho_dU_{k}^1}{\rho_d\sum_{k_1\ne k}U_{kk_1}^2+\rho_d\sum_{k_1 \in \mathcal{K} }U_{kk_1}^3+ z_k},
\end{equation}
where $\text{DS}_k$ and $\text{IN}_k$ represent the DS and IN for user $k$. The expression of $\text{DS}_k$ and $\text{IN}_k$ is given as follows:
\begin{equation}\label{expression of TS and IN}
  \begin{aligned}
    \text{DS}_k &= \rho_dU_{k}^1,\\
    \text{IN}_k &= \rho_d\sum_{k_1\ne k}U_{kk_1}^2+\rho_d\sum_{k_1 \in \mathcal{K} }U_{kk_1}^3+z_k,
  \end{aligned}
\end{equation}
The expression of $U_{k}^1$, $U_{kk_1}^2$ and $U_{kk_1}^3$ are given by
\begin{subequations}\label{expression of Ukk}
  \begin{align}
    &U_{k}^1 = \left |\sum_{m\in \mathcal{M}}\sqrt{\eta_{mk}} \hat{\mathbf{h}}_{mk}^H\mathbf{b}_{mk}\right |^2, \\
    &U_{kk_1}^2 = \left |\sum_{m\in \mathcal{M}}\sqrt{\eta_{mk_1}}  \hat{\mathbf{h}}_{mk}^H\mathbf{b}_{mk_1}\right |^2 , \\
    &U_{kk_1}^3 = \left |\sum_{m\in \mathcal{M}}\sqrt{\eta_{mk_1}}  \bar{\mathbf{h}}_{mk}^H\mathbf{b}_{mk_1}\right |^2.
  \end{align}
\end{subequations}
For the convenience of the following analysis, we redefine $U_{k}^1 = \left |\xi_{mk}\right |^2$, $U_{k}^2 = \left |\xi_{mkk_1}\right |^2$ and $U_{k}^3 = \left |\psi_{mkk_1}\right |^2$, where $\xi_{mk}$, $\xi_{mkk_1}$ and $\psi_{mkk_1}$ is given by $\xi_{mk} =\sqrt{\eta_{mk}} \hat{\mathbf{h}}_{mk}^H\mathbf{b}_{mk}$, $\xi_{mkk_1} = \sqrt{\eta_{mk_1}}  \hat{\mathbf{h}}_{mk}^H\mathbf{b}_{mk_1}$ and $\psi_{mkk_1} = \sqrt{\eta_{mk_1}}  \bar{\mathbf{h}}_{mk}^H\mathbf{b}_{mk_1}$, respectively.

\subsection{Maximum Ratio Transmission}
In this subsection, we derive the distribution of SINR when MRT precoding is used.
In the presence of multiple users receiving signals from all APs simultaneously, it becomes evident from equation (\ref{uplink SINR exprssion}) that all signals are interconnected, giving rise to the formation of the DS and the IN. Consequently, directly analyzing the distribution of SINR under MRT precoding is a challenging task. To overcome this difficulty, we adopt a feasible approach in the subsequent analysis. Firstly, we scrutinize the expressions for the distribution of DS and IN individually. Subsequently, by utilizing the derived expressions for the distributions of DS and IN, we are able to obtain the distribution of SINR.

When the number of APs $M$ and number of antennas at each AP $N$ are large enough, the distribution of $\sum_{m\in \mathcal{M}}\sqrt{\eta_{k}} \mathbf{b}_{mk}^H\hat{\mathbf{h}}_{mk}$ is approximated as a Gaussian distribution. Then we can approximate $U_k^1$ as a Gamma distribution in the following lemma.

\begin{lemma}\label{lemma of U1 downlink MRT}
  When MRT precoding is used for downlink data transmission, the distribution of $U_{k}^1$ can be approximated as a Gamma distribution with the shape parameter $j_{k1}$ and the scale parameter $\chi_{k1}$, i.e., $U_{k}^1 \sim Gamma\left ( j_{k1},\chi_{k1} \right )$. Therefore, the PDF and cumulative distribution function (CDF) of $U_{k}^1$ is expressed as follows:
  \begin{equation}\label{PDF of gamma U}
    \left \{
    \begin{aligned}
      &f_{U_{k}^1}(x) = \frac{1}{\Gamma(j_{k1})\chi_{k1}^{j_{k1}}}x^{j_{k1}-1}e^{-\frac{x}{\chi_{k1}} },\\
      &F_{U_{k}^1}(x) = \frac{1}{\Gamma(j_{k1})}\bar{\gamma}\left (j_{k1} ,\frac{x}{\chi_{k1}} \right ),
    \end{aligned}
     \right.
  \end{equation}
  where $\Gamma(z) = \int_{0}^{+\infty} t^{z-1} e^{-t} dt$ and $\bar{\gamma}(a,z) = \int_{0}^{z} t^{a-1} e^{-t} dt$ represent the gamma function and incomplete gamma function, respectively. The shape parameter $j_{k1}$ and scale parameter $\chi_{k1}$ are given by
  \begin{equation}\label{gamma paraments U1 downlink MRT}
      j_{k1} = \frac{u_{U_{k}^1}^2}{u_{U_{k}^1}^{(2)} - u_{U_{k}^1}^2},\ \chi_{k1} = \frac{u_{U_{k}^1}^{(2)} - u_{U_{k}^1}^2}{u_{U_{k}^1}},
  \end{equation}
  where $u_{U_{k}^1} = \mathbb{E}\left \{ U_{k}^1 \right \} $ and $u_{U_{k}^1}^{(2)} = \mathbb{E}\left \{  \left ( U_{k}^1  \right ) ^2 \right \}$ denote the first and second moments of $U_{k}^1$ respectively. The expressions for $u_{U_{k}^1}$ and $u_{U_{k}^1}^{(2)}$ are given in (\ref{first order of U1 downlink MRT}) and (\ref{second order of U1 downlink MRT}), respectively.
\end{lemma}

\begin{proof}
  Please refer to Appendix \ref{appendix A}.
\end{proof}

Since $\text{DS}_k = \rho_d U_k^1$, we have $\mathbb{E} \left \{  \text{DS}_k \right \} = \rho_d \mathbb{E} \left \{  U_k^1  \right \} $ and $\mathbb{E} \left \{  \text{DS}_k^2 \right \} = \rho_d \mathbb{E} \left \{  \left ( U_k^1  \right )^2  \right \}$. Then, $\text{DS}_k^{\text{MRT}}$ can be approximated as a Gamma distribution with shape parameter $j_{k1}$ and scale parameter $\rho_d\chi_{k1}$, i.e., $\text{DS}_k^{\text{MRT}} \sim \text{Gamma}(j_{k1},\rho_d\chi_{k1} )$, where $\text{DS}_k^{\text{MRT}}$ represents the $\text{DS}_k$ with MRT precoding.

Then we turn to analyze the distribution of $\text{IN}_k^{\text{MRT}}$, where $\text{IN}_k^{\text{MRT}}$ represents the $\text{IN}_k$ with MRT precoding. Similarly, the distribution of $\text{IN}_k^{\text{MRT}}$ can be approximated as a Gamma distribution. In the following lemma, we approximate $\text{IN}_k^{\text{MRT}}$ as a Gamma distribution.

\begin{lemma}\label{lemma of IN downlink MRT}
  When MRT precoding is used for downlink data transmission, the distribution of $\text{IN}_k^{\text{MRT}} = \rho_d\sum_{k_1\ne k}U_{kk_1}^2+\rho_d\sum_{k_1 \in \mathcal{K} }U_{kk_1}^3+z_k$ can be approximated as a Gamma distribution with the shape parameter $j_{k2}$ and the scale parameter $\chi_{k2}$, i.e., $\text{IN}_k^{\text{MRT}} \sim \text{Gamma}(j_{k2},\chi_{k2})$. The shape parameter $j_{k2}$ and scale parameter $\chi_{k2}$ are given by
  \begin{equation}\label{gamma paraments IN downlink MRT}
      j_{k2} = \frac{u_{\text{IN}_{k}^{\text{MRT}}}^2}{u_{\text{IN}_{k}^{\text{MRT}}}^{(2)} - u_{\text{IN}_{k}^{\text{MRT}}}^2} ,\ \chi_{k2} = \frac{u_{\text{IN}_{k}^{\text{MRT}}}^{(2)} - u_{\text{IN}_{k}^{\text{MRT}}}^2}{u_{\text{IN}_{k}^{\text{MRT}}}},
  \end{equation}
  where $u_{\text{IN}_{k}^{\text{MRT}}} = \mathbb{E}\left \{\text{IN}_{k}^{\text{MRT}}\right \} $ and $u_{\text{IN}_{k}^{\text{MRT}}}^{(2)} = \mathbb{E}\left \{ \left (\text{IN}_{k}^{\text{MRT}}\right )^2 \right \}$ are first and second moments of $\text{IN}_k^{\text{MRT}}$, respectively. The expressions for $u_{\text{IN}_{k}^{\text{MRT}}}$ and $u_{\text{IN}_{k}^{\text{MRT}}}^{(2)}$ are given in (\ref{fisrst order of IN downlink MRT}) and (\ref{second order of IN downlink MRT}), respectively.
\end{lemma}

\begin{proof}
  Please refer to Appendix \ref{appendix B}.
\end{proof}
Based on the above analysis, we have obtained the distributions of the two components of the SINR: DS and IN. By leveraging the independence property between these two components, we can obtain the distribution of SINR in the CF mMIMO system which is shown in the following theorem.

\begin{theorem}\label{PDF of gamma when MRT downlink}
  In the CF mMIMO system with pilot contamination, when MRT precoding is used for downlink data transmission, the PDF and CDF of the SINR for user $k$, $\forall k \in \mathcal{K}$ can be expressed as follows:
  \begin{equation}\label{pdf of gamma downlink MRT}
    \left \{
    \begin{aligned}
      f_{\gamma_k}^{\text{MRT}}(x) =& \frac{\Gamma(j_{k1}+j_{k2})x^{j_{k1}-1}(\frac{1}{\chi_{k2}}+\frac{x}{\rho_d\chi_{k1}}) ^{-j_{k1}-j_{k2}}}{\Gamma(j_{k1})\Gamma(j_{k2})\left ( \rho_d\chi_{k1} \right ) ^{j_{k1}}\chi_{k2}^{j_{k2}}},\\
      F_{\gamma_k}^{\text{MRT}}(x) =& \frac{\Gamma(j_{k1}+j_{k2}) \chi_{k2}^{j_{k1}} x^{j_{k1}}}{j_{k1}\Gamma(j_{k1})\Gamma(j_{k2})\left ( \rho_d\chi_{k1} \right )^{j_{k1}}}\\
                                    &\times H(j_{k1},j_{k1}+j_{k2},j_{k1}+1,-\frac{\chi_{k2}x}{\rho_d\chi_{k1}}),
    \end{aligned}
     \right .
  \end{equation}
  where $j_{k1}$, $\chi_{k1}$, $j_{k2}$, and $\chi_{k2}$ are given in Lemma \ref{lemma of U1 downlink MRT} and Lemma \ref{lemma of IN downlink MRT}, respectively, $H(a,b,c)$ is the hypergeometric function \cite{HypergeometricFunction}.
\end{theorem}

\begin{proof}
  Please refer to Appendix \ref{appendix C}.
\end{proof}

\subsection{Full-Pilot Zero-Forcing}
Different from MRT precoding, FZF precoding can suppress inter-user interference. The FZF precoding vector used by AP $m$ towards user $k$ is given in (\ref{proceding vector}), and the denominator of the precoding vector is given in closed form by \cite{InterdonatoLocal2020}
\begin{equation}\label{close form of precoder vector}
  \begin{aligned}
    \mathbb{E} \left \{ \left \| \bar{\mathbf{H}}_{m}\left [ \bar{\mathbf{H}}_{m}^H\bar{\mathbf{H}}_{m} \right ]^{-1} \mathbf{e}_{i_k}\right \|^2\right \}  &=\mathbb{E}\left \{ \left [ \left ( \bar{\mathbf{H}}_{m}^H\bar{\mathbf{H}}_{m} \right )^{-1} \right ]_{i_ki_k}  \right \}  \\
    &\overset{(a)}{=}  \frac{\kappa_{mk}^2}{(N-l_p)c_{mk}},
  \end{aligned}
\end{equation}
where $(a)$ is obtained based on \cite[Lemma 2.10]{edelman_rao_2005}, for a $l_p \times l_p$ central complex Wishart matrix with $M$ degrees of freedom satisfying $M \ge l_p +1$.
Interference between users using different pilot sequences is suppressed and the product between $\hat{\mathbf{h}}_{mk}^H$ and $\mathbf{b}_{mk'}$ can be calculated as follows \cite{InterdonatoLocal2020}:
\begin{equation}\label{channel times precoding vextor in FZF}
  \begin{aligned}
   &\alpha_{mkk_1}=\hat{\mathbf{h}}_{mk}^H\mathbf{b}_{mk_1} \\
   & = (\kappa_{mk}\bar{\mathbf{H}}_m\mathbf{e}_{i_k})^H\bar{\mathbf{H}}_{m}\left (\bar{\mathbf{H}}_{m}^H\bar{\mathbf{H}}_{m} \right )^{-1}\mathbf{e}_{i_{k_1}}\sqrt{\frac{(N-l_p)c_{mk}}{\kappa_{mk}^2}} \\
   & = \left \{
       \begin{aligned}
          & 0, && k_1 \notin \mathcal{P}_k,\\
          &\sqrt{(N-l_p)c_{mk}}, && k_1 \in \mathcal{P}_k.
       \end{aligned}
        \right .
  \end{aligned}
\end{equation}
Then we can obtain the value of $\text{DS}_k$ based on (\ref{channel times precoding vextor in FZF}) as follows:
\begin{equation}\label{expression of Uk1 in FZF}
  \begin{aligned}
     \text{DS}_k^{\text{FZF}} &= \rho_d\left ( {\textstyle \sum_{m\in \mathcal{M}}}\sqrt{\eta_{mk}}\alpha_{mkk}\right )^2 \\
                              &= \rho_d\left ({\textstyle \sum_{m\in \mathcal{M}}} \sqrt{\eta_{mk}(N-l_p)c_{mk}}\right )^2,
  \end{aligned}
\end{equation}
where $\text{DS}_k^{\text{FZF}}$ represents the $\text{DS}_k$ with FZF precoding.

Similar to the analysis under MRT precoding, we approximate the distribution of $\text{IN}_k^{\text{FZF}}$ under FZF precoding as a Gamma distribution in the following lemma.
\begin{lemma}\label{lemma of IN downlink FZF}
  When FZF precoding is used for downlink data transmission, the distribution of $\text{IN}_k^{\text{FZF}} = \rho_d\sum_{k_1\ne k}U_{kk_1}^2+\rho_d\sum_{k_1 \in \mathcal{K} }U_{kk_1}^3+z_k$ can be approximated as a Gamma distribution with the shape parameter $j_{k2}$ and the scale parameter $\chi_{k2}$, i.e., $\text{IN}_k^{\text{FZF}} \sim Gamma(j_{k2},\chi_{k2})$. The shape parameter $j_{k2}$ and scale parameter $\chi_{k2}$ are given by
  \begin{equation}\label{gamma paraments IN downlink FZF}
      j_{k2} = \frac{u_{\text{IN}_k^{\text{FZF}}}^2}{u_{\text{IN}_k^{\text{FZF}}}^{(2)} - u_{\text{IN}_k^{\text{FZF}}}^2} ,\ \chi_{k2} = \frac{u_{\text{IN}_k^{\text{FZF}}}^{(2)} - u_{\text{IN}_k^{\text{FZF}}}^2}{u_{\text{IN}_k^{\text{FZF}}}},
  \end{equation}
  where $u_{\text{IN}_k^{\text{FZF}}} = \mathbb{E}\left \{ \text{IN}_k^{\text{FZF}} \right \} $ and $u_{\text{IN}_k^{\text{FZF}}}^{(2)} = \mathbb{E}\left \{ \left ( \text{IN}_k^{\text{FZF}} \right )^2 \right \}$ are first and second moments of $\text{IN}_k^{\text{FZF}}$, respectively. The expressions for $u_{\text{IN}_k^{\text{FZF}}}$ and $u_{\text{IN}_k^{\text{FZF}}}^{(2)}$ are given in (\ref{fisrst order of IN downlink FZF}) and (\ref{second order of IN downlink FZF}), respectively.
\end{lemma}

\begin{proof}
  Define $U_{k}^2 = \rho_d\sum_{k_1\ne k}U_{kk_1}^2$, similar to (\ref{expression of Uk1 in FZF}), the value of $U_{k}^2$ can be expressed as follows:
  \begin{equation}\label{expression of Uk2 in FZF}
    U_{k}^2 = \rho_d\sum_{k_1\ne k} \left ( \sum_{m\in \mathcal{M}}\sqrt{\eta_{mk_1}} \alpha_{mkk_1} \right )^2.
  \end{equation}
  Similar to (\ref{second and fourth of psi downlink MRT}), the second and fourth moment of $\psi^d_{mkk_1}$ can be calculated as follows:
  \begin{equation}\label{second and fourth of psi downlink FZF}
    \begin{aligned}
      &\mathbb{E}\left \{ \left | \psi_{mkk_1} \right |^2  \right \}  = \eta_{mk_1} \left ( \beta_{mk}-c_{mk} \right )\\
      &\mathbb{E}\left \{ \left | \psi_{mkk_1} \right |^4  \right \}  = \eta_{mk_1}^2\frac{2\left ( N+1 \right ) }{N} \left ( \beta_{mk}-c_{mk} \right )^2\\
    \end{aligned}
  \end{equation}
  Then we can obtain the first moment and second moment of $U_{kk_1}^3$ under FZF precoding, i.e., $\mathbb{E}\left \{ U_{kk_1}^3 \right \}$, $\mathbb{E}\left \{ \left ( U_{kk_1}^3 \right )^2 \right \}$, which is similar to (\ref{first and second moment of Uk3}) in the proof of Lemma \ref{lemma of IN downlink MRT}.
  Then first moment of $\text{IN}_k^{\text{FZF}}$ can be expressed as follows:
  \begin{equation}\label{fisrst order of IN downlink FZF}
    u_{\text{IN}_k^{\text{FZF}}} = U_{k}^2 +  \rho_d\sum_{k_1\in \mathcal{K} }\mathbb{E}\left \{ U_{kk_1}^3 \right \}  + 1.
  \end{equation}
  The second moments of $\text{IN}_k^{\text{FZF}}$ can be calculated as follows:
  \begin{equation}\label{second order of IN downlink FZF}
    \begin{aligned}
      &u_{\text{IN}_k^{\text{FZF}}}^{(2)} \\
      &= \rho_d^2\sum_{k_1\in \mathcal{K} } \mathbb{E}\left \{ \left ( U_{kk_1}^3 \right )^2 \right \} +2\rho_d\left (  U_{k}^2 + 1 \right ) \sum_{k_1\in \mathcal{K} }\mathbb{E}\left \{ U_{kk_1}^3 \right \}  \\
      &  + \rho_d^2\sum_{k_1\in \mathcal{K} }\sum_{k_2\ne k_1}\mathbb{E}\left \{ U_{kk_1}^3 \right \}\mathbb{E}\left \{  U_{kk_2}^3 \right \}+ \left ( U_{k}^2 \right )^2  + 2 U_{k}^2 + 2 .
    \end{aligned}
  \end{equation}
\end{proof}
We calculate the value of $\text{DS}_k^{\text{FZF}}$ and obtain the distribution of $\text{IN}_k^{\text{FZF}}$ above. Then the distribution of SINR under the  FZF precoding can be obtained as shown in the following theorem.
\begin{theorem}\label{PDF of gamma when FZF downlink}
  In the CF mMIMO system with pilot contamination, when FZF precoding is used for downlink data transmission, the PDF and CDF of the SINR for user $k$, $\forall k \in \mathcal{K}$ can be expressed as follows:
  \begin{equation}\label{pdf of gamma downlink FZF}
    \begin{aligned}
      &f_{\gamma_k}^{\text{FZF}}(x) = \frac{1}{\Gamma(j_{k2})\chi_{k2}^{j_{k2}}x}\left ( \frac{\text{DS}_k^{\text{FZF}}}{x}  \right ) ^{j_{k2}}e^{-\frac{\text{DS}_k^{\text{FZF}}}{x\chi_{k2}}}, \\
      &F_{\gamma_k}^{\text{FZF}}(x) = 1-\frac{1}{\Gamma(j_{k2})}\bar{\gamma}\left (j_{k2} ,\frac{\text{DS}_k^{\text{FZF}}}{x\chi_{k2}} \right ),
    \end{aligned}
  \end{equation}
  where $\text{DS}_k^{\text{FZF}}$ is given in (\ref{expression of Uk1 in FZF}), $j_{k2}$, and $\chi_{k2}$ are given in Lemma \ref{lemma of IN downlink FZF}.
\end{theorem}

\begin{proof}
  The CDF of $\gamma_k$ when using the FZF precoder can be calculated as follows:
  \begin{equation}\label{CDF of gammak FZF}
    \begin{aligned}
      F_{\gamma_k}^{\text{FZF}}(x) &= P\left \{ \frac{\text{DS}_k^{\text{FZF}}}{\text{IN}_k^{\text{FZF}}}\le x\right \} = P\left \{ \text{IN}_k^{\text{FZF}} \ge\frac{\text{DS}_k^{\text{FZF}}}{x} \right \} \\
      & = 1 - F_{\text{IN}_k^{\text{FZF}}}\left ( \frac{\text{DS}_k^{\text{FZF}}}{x} \right ),
    \end{aligned}
  \end{equation}
  where $F_{\text{IN}_k^{\text{FZF}}}\left ( \cdot \right )$ represents the CDF of $\text{IN}_k^{\text{FZF}}$.

  Based on (\ref{CDF of gammak FZF}), the PDF of $\gamma_k$ when using FZF precoder can be calculated as follows:
  \begin{equation}\label{PDF of gammak FZF}
    \begin{aligned}
      f_{\gamma_k}^{\text{FZF}}(x) &= \frac{\mathrm{d} F_{\gamma_k}^{\text{FZF}}(x)}{\mathrm{d} x} = -\frac{\mathrm{d} F_{\text{IN}_k}^{\text{FZF}}(\frac{\text{DS}_k^{\text{FZF}}}{x})}{\mathrm{d} x}\\
                      &= \frac{\text{DS}_k^{\text{FZF}}}{x^2}f_{\text{IN}_k^{\text{FZF}}}\left ( \frac{\text{DS}_k^{\text{FZF}}}{x}  \right ),
    \end{aligned}
  \end{equation}
  where $f_{\text{IN}_k^{\text{FZF}}}\left ( \cdot \right )$ represents the PDF of $\text{IN}_k^{\text{FZF}}$.
\end{proof}

\section{Performance Analysis}
The lower bound of achievable rate for users in the CF mMIMO system has been investigated \cite{InterdonatoLocal2020}.
However, to provide a more comprehensive characterization of the system's performance, this section focuses on deriving the achievable rate of users under both MRT and FZF precoding. This analysis is based on the previously obtained distribution of SINR, allowing for a more accurate assessment of system performance.
Besides, we also derive the outage probability of the CF mMIMO system.
\subsection{Maximum Rate Transmission}
When MRT precoding is used in the CF mMIMO system, then the lower bound of achievable downlink rate is given in (\ref{LB of AR MRT}) \cite{InterdonatoLocal2020}.
The achievable rate in the CF mMIMO system with MRT precoding employed is given in the following lemma.
\begin{lemma}\label{lemma of AR MRT}
  In the CF mMIMO system with pilot contamination, when MRT precoding is employed for downlink data transmission, the achievable rate for user $k$, $k \in \mathcal{K}$, is given in (\ref{achievable rate MRT}),
  where $H^P\left ( \left \{ \cdot \right \},\left \{ \cdot  \right \},\left \{ \cdot  \right \}   \right ) $ represents the generalized hypergeometric function \cite{GeneralizedHypergeometricFunction}. The values of $j_{k1}$ and $\chi_{k1}$ are obtained from Lemma \ref{lemma of U1 downlink MRT}, while $j_{k2}$ and $\chi_{k2}$ are derived from Lemma \ref{lemma of IN downlink MRT}.
\end{lemma}

\begin{proof}
  When MRT precoding is employed for downlink data transmission, the PDF of SINR for user $k$, $k \in \mathcal{K}$, is given in Theorem \ref{PDF of gamma when MRT downlink}. Then the achievable rate for user $k$ can be calculated directly by $R_k^{\text{MRT}} = \int_{0}^{\infty} \log_2(1+x)f_{\gamma_k}^{\text{MRT}}(x)dx$.
\end{proof}

Consider the encoding rate of user $k$ is $r_k^{\text{MRT}}$ in the CF mMIMO system unser MRT precoding. An outage event for user $k$ when the SINR cant support the target rate $r_k^{\text{MRT}}$. The outage probability of user $k$ in the CF mMIMO system is given by
\begin{equation}\label{outage probability MRT}
  \begin{aligned}
    P_{\text{out},k}^{\text{MRT}}\left ( r_k^{\text{MRT}} \right ) &= P\left ( \log_2\left ( 1 + \gamma_k \le  r_k^{\text{MRT}}\right ) \right )\\
                                                                   &=F_{\gamma_k}^{\text{MRT}}\left ( 2^{r_k^{\text{MRT}}}  - 1\right ).
   \end{aligned}
\end{equation}

\subsection{Full-Pilot Zero-Forcing}
When FZF precoding is used in the CF mMIMO system, then the lower bound of achievable downlink rate is given (\ref{LB of AR FZF}) \cite{InterdonatoLocal2020}.

The achievable rate in the CF mMIMO system with FZF precoding employed is given in the following lemma.

\begin{lemma}\label{lemma of AR FZF}
  In the CF mMIMO system with pilot contamination, when FZF precoding is employed for downlink data transmission, the achievable rate for user $k$, $k \in \mathcal{K}$, is given in (\ref{achievable rate FZF}), where $\text{DS}_k^{\text{FZF}}$ is given in (\ref{expression of Uk1 in FZF}), while $j_{k2}$ and $\chi_{k2}$ are obtained from Lemma \ref{lemma of IN downlink FZF}.
\end{lemma}

\begin{proof}
  The proof is similar to the proof of Lemma \ref{lemma of AR MRT}, which is omitted for simplicity.
\end{proof}

Consider the encoding rate of user $k$ is $r_k^{\text{MRT}}$ in the CF mMIMO system under FZF precoding. The outage probability of user $k$ is given by
\begin{equation}\label{outage probability FZF}
  \begin{aligned}
    P_{\text{out},k}^{\text{MRT}}\left ( r_k^{\text{MRT}} \right ) =F_{\gamma_k}^{\text{FZF}}\left ( 2^{r_k^{\text{FZF}}}  - 1\right ).
   \end{aligned}
\end{equation}

\begin{figure*}
    \begin{equation}\label{LB of AR MRT}
      \bar{R}_k^{\text{MRT}} = \log_2\left ( 1 + \frac{N\rho_d\left ( \sum_{m \in \mathcal{M}}\sqrt{\eta_{mk}c_{mk}}\right )^2}{N\rho_d\sum_{k_1 \in \mathcal{P}_k \setminus \left \{ k \right \} }\left ( \sum_{m \in \mathcal{M}}\sqrt{\eta_{mk}c_{mk}}\right )^2 + \rho_d\sum_{m \in \mathcal{M} }\sum_{k}^{\mathcal{K} } \eta_{mk}\beta _{mk} +1}  \right ).
    \end{equation}
    \begin{equation}\label{LB of AR FZF}
      \begin{aligned}
       &\bar{R}_k^{\text{FZF}} = \log_2\left ( 1 + \frac{\left ( N-l_p \right ) \rho_d\left ( \sum_{m \in \mathcal{M}}\sqrt{\eta_{mk}c_{mk}}\right )^2}{\left ( N-l_p \right )\rho_d\sum_{k_1 \in \mathcal{P}_k \setminus \left \{ k \right \} }\left ( \sum_{m \in \mathcal{M}}\sqrt{\eta_{mk}c_{mk}}\right )^2 + \rho_d\sum_{m \in \mathcal{M} }\sum_{k}^{\mathcal{K} } \eta_{mk}\beta _{mk} +1}  \right ).
      \end{aligned}
    \end{equation}
	{\noindent} \rule[-10pt]{18cm}{0.05em}
\end{figure*}

\begin{figure*}
    \begin{equation}\label{achievable rate MRT}
      \begin{aligned}
        R_k^{\text{MRT}} =& \frac{\pi \csc(j_{k2}\pi)}{ \Gamma(j_{k1})  \Gamma(j_{k2}) \ln2 } \left ( \frac{\left ( \rho_d\chi_{k1} \right ) ^{j_{k2}}\Gamma(j_{k1}+j_{k2})}{\chi_{k2}^{j_{k2}}\Gamma(1+j_{k2})} H(j_{k2},j_{k1}+j_{k2},1+j_{k2},\frac{ \rho_d\chi_{k1} }{\chi_{k2}} )   \right .   \\
        &  \left .- \frac{\rho_d\chi_{k1} \Gamma(1+j_{k1})}{\chi_{k2}\Gamma(2-j_{k2})}H^P\left ( \left \{ 1,1,1+j_{k1} \right \},\left \{ 2,2-j_{k2} \right \},\frac{\rho_d\chi_{k1}  }{\chi_{k2}}   \right )   \right ).
      \end{aligned}
    \end{equation}
    \begin{equation}\label{achievable rate FZF}
      \begin{aligned}
        R_k^{\text{FZF}} = & \frac{\left ( -1\right )^{-j_{k2}}\pi \csc(j_{k2}\pi)}{\Gamma(j_{k2})\ln2}  \bar{\gamma}\left ( -\frac{\text{DS}_k^{\text{FZF}}}{\chi_{k2}} ,j_{k2} \right ) +\frac{\text{DS}_k^{\text{FZF}}\Gamma(-1 + j_{k2}) }{\chi_{k2}\Gamma(j_{k2})\ln2} H^P\left (  \left \{ 1,1 \right \} ,\left \{ 2,2-j_{k2} \right \} , \frac{\text{DS}_k^{\text{FZF}}}{\chi_{k2}}\right ) .
      \end{aligned}
    \end{equation}
	{\noindent} \rule[-10pt]{18cm}{0.05em}
\end{figure*}

\section{Simulation Results}

In this section, we validate our derived results by conducting Monte Carlo simulations across various scenarios and subsequently perform corresponding performance analysis.

\subsection{Simulation Setting}
In our simulations, we consider a system of randomly distributed APs and users within a rectangular area measuring $1km \times 1km$. The large-scale fading coefficient, which accounts for both path loss and shadowing effects, is denoted as $\beta_{mn} = PL_{mn} + z_{mn}$, where $PL_{mn}$ represents the path loss component, while $z_{mn} \sim \mathcal{CN}(0,\delta_{sh}^2)$ represents the shadowing component following a complex Gaussian distribution with a mean of zero and a variance of $\delta_{sh}^2$. To characterize the path loss, we employ a three-slope model proposed in \cite{NgoCell2017} and adopt the same parameter settings as described in that paper.
In addition, we set the pilot power to 20 dBm and the downlink transmit power to 23 dBm.
The system bandwidth is 2 MHz and the power of Gaussian noise is -174 dB/Hz.
The pilots are sequentially assigned to users, i.e., $i_k = \text{rem}\left ( k,l_p \right )$, where $\text{rem}\left ( ,\right )$ represents the remainder operation.
We employ the heuristic power allocation scheme introduced in \cite{InterdonatoLocal2020}, and the power allocation coefficients are computed using the following equation:
\begin{equation}\label{power allocation}
  \eta_{mk} = \frac{c_{mk}}{ {\textstyle \sum_{k \in \mathcal{K} } c_{mk}} }, \forall m,\forall k.
\end{equation}
For the specific configuration details regarding the number of APs, users, the number of antennas per AP, and the pilot length in the system, we will provide them during the corresponding experiments.

\subsection{CDF of SINR}

\begin{figure}[ht]
	\centering
    \setlength{\abovecaptionskip}{-0.1cm}   
	\includegraphics[scale=0.35]{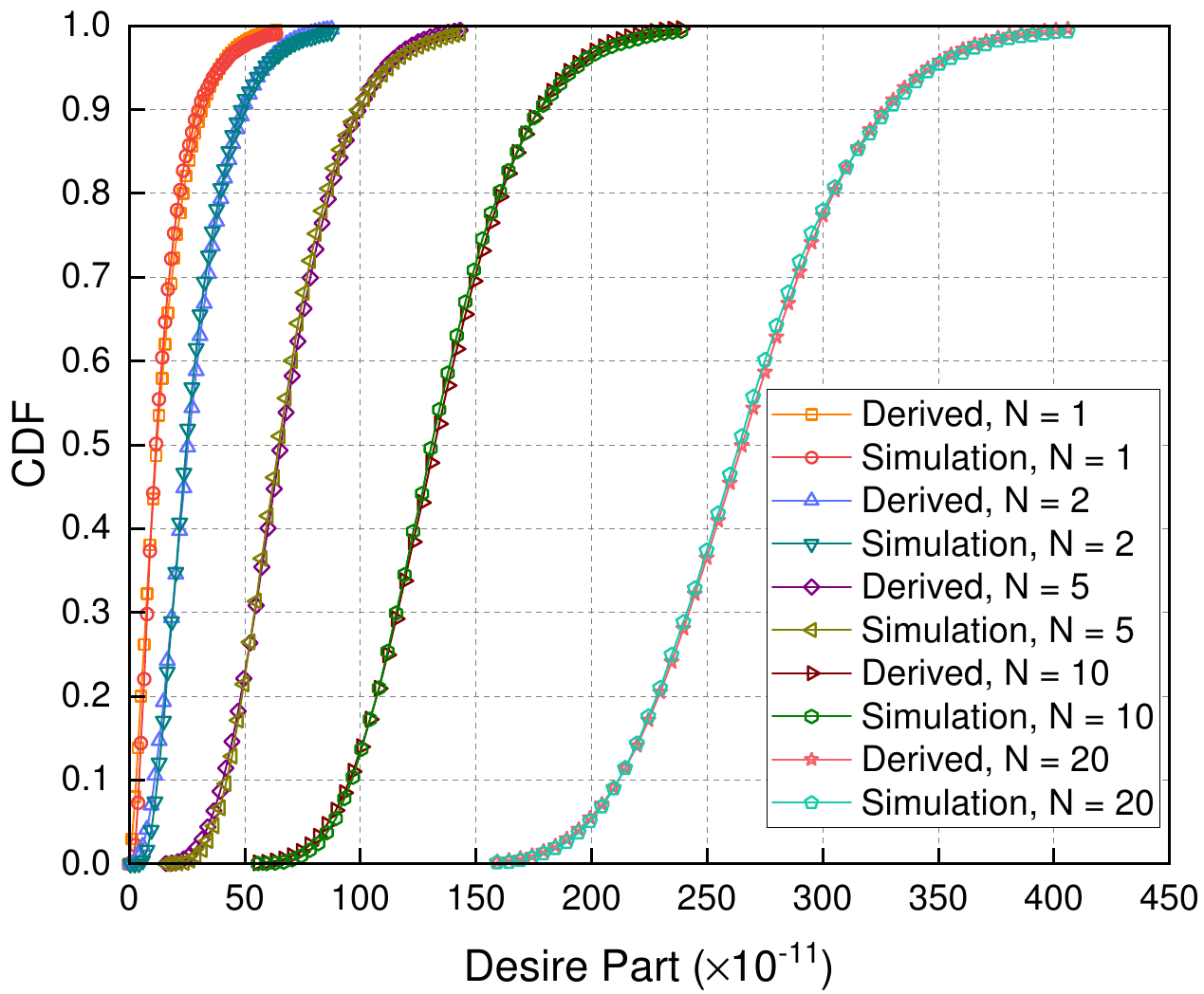}
	\caption{CDF of DS under MRT precoding with different number of antennas each AP. System parameters: $M$ = 120, $K$ = 20, $l_p$ = 10.}
	\label{MRTDSgraph}
\end{figure}

\begin{figure}[ht]
	\centering
    \setlength{\abovecaptionskip}{-0.1cm}   
	\includegraphics[scale=0.35]{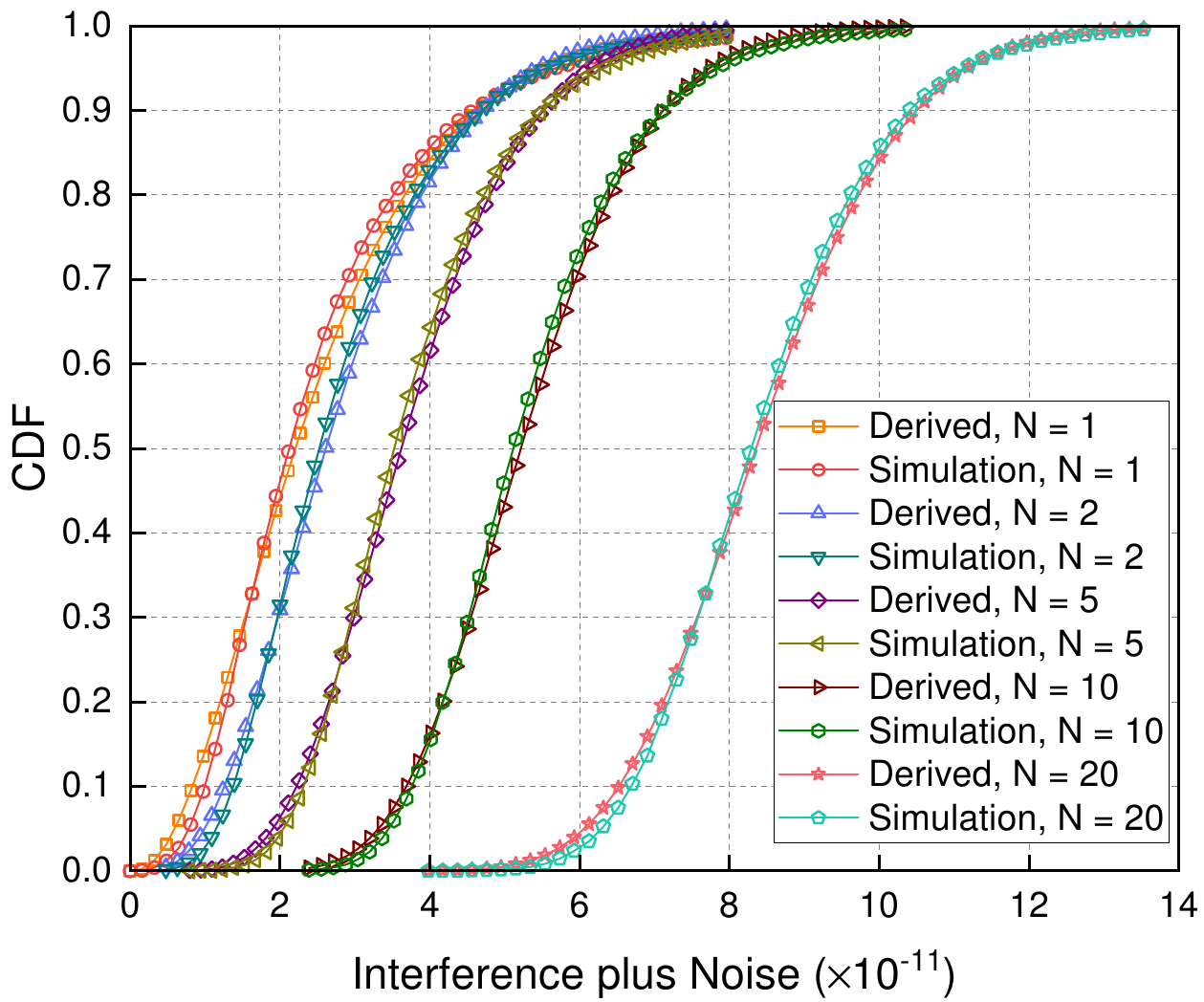}
    \caption{CDF of IN under MRT precoding with different number of antennas each AP. System parameters: $M$ = 120, $K$ = 20, $l_p$ = 10.}
	\label{MRTINgraph}
\end{figure}

Once MRT precoding is implemented in CF mMIMO systems, the CDF of DS and IN, concerning the variation in the number of antennas at the AP, are illustrated in Fig. \ref{MRTDSgraph} and Fig. \ref{MRTINgraph}, respectively. As anticipated, both DS and IN exhibit a corresponding increase with the augmentation of antenna quantity. This phenomenon can be attributed to the fact that a higher number of antennas enables the AP to transmit a greater number of signals, consequently leading to an enhancement in signal strength.
Furthermore, the results derived from our theoretical analysis align closely with the outcomes obtained through Monte Carlo simulations. This congruence serves as evidence supporting the accuracy of our conclusions.

\begin{figure}[ht]
	\centering
    \setlength{\abovecaptionskip}{-0.1cm}   
	\includegraphics[scale=0.35]{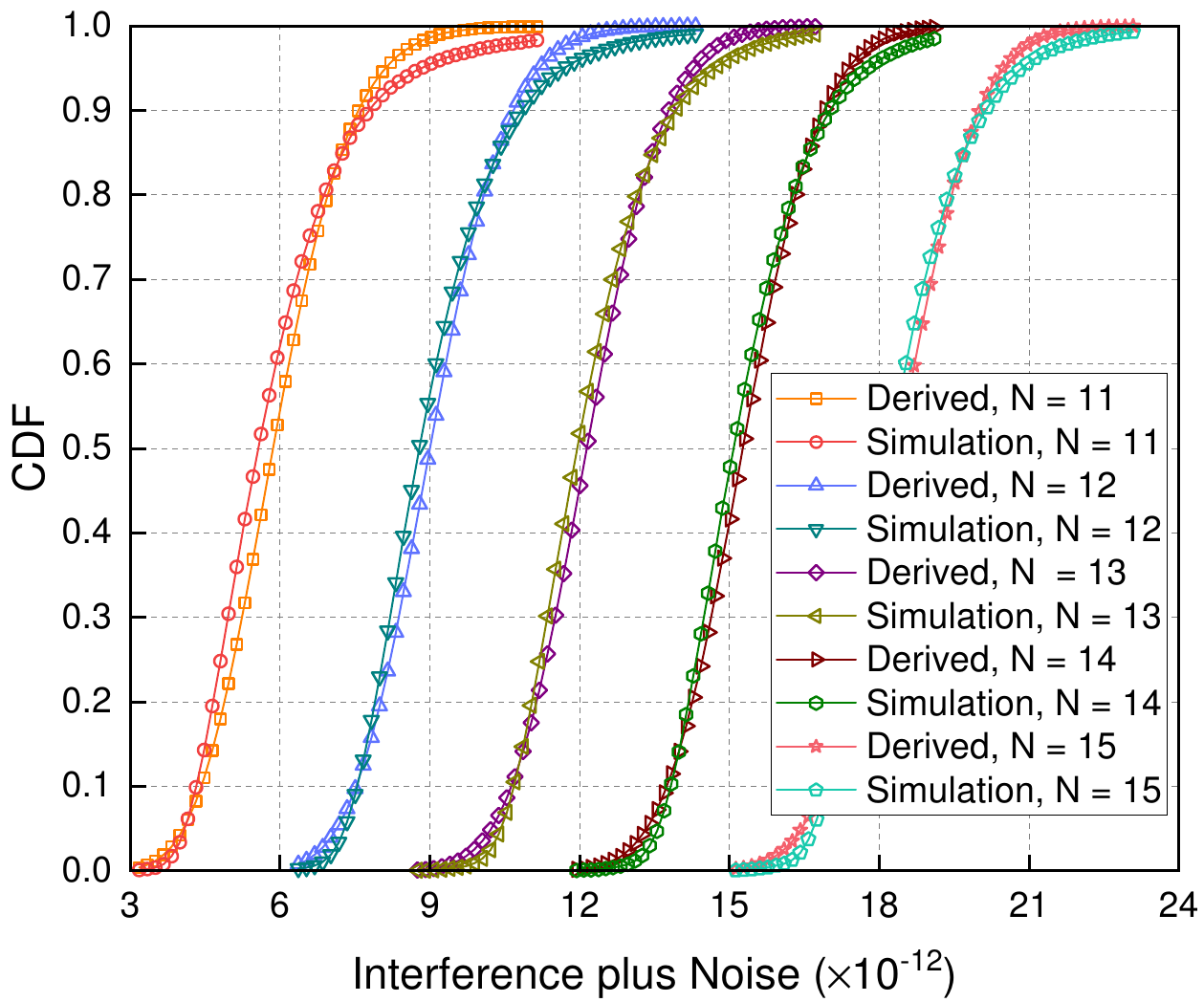}
	\caption{CDF of IN under FZF precoding with different number of antennas each AP. System parameters: $M$ = 120, $K$ = 20, $l_p$ = 10.}
	\label{FZFINgraph}
\end{figure}

In Fig. \ref{FZFINgraph}, we investigate the impact of the variation in the number of antennas at each AP on the CDF of IN when FZF precoding is deployed. As the number of antennas increases at each AP, interference becomes more severe, resulting in a proportional increase of IN. Moreover, the graph illustrates a close alignment between the distribution derived from our theoretical analysis and the distribution obtained through Monte Carlo simulations, providing strong evidence that our conclusions are accurate.

\begin{figure}[ht]
	\centering
    \setlength{\abovecaptionskip}{-0.1cm}   
	\includegraphics[scale=0.35]{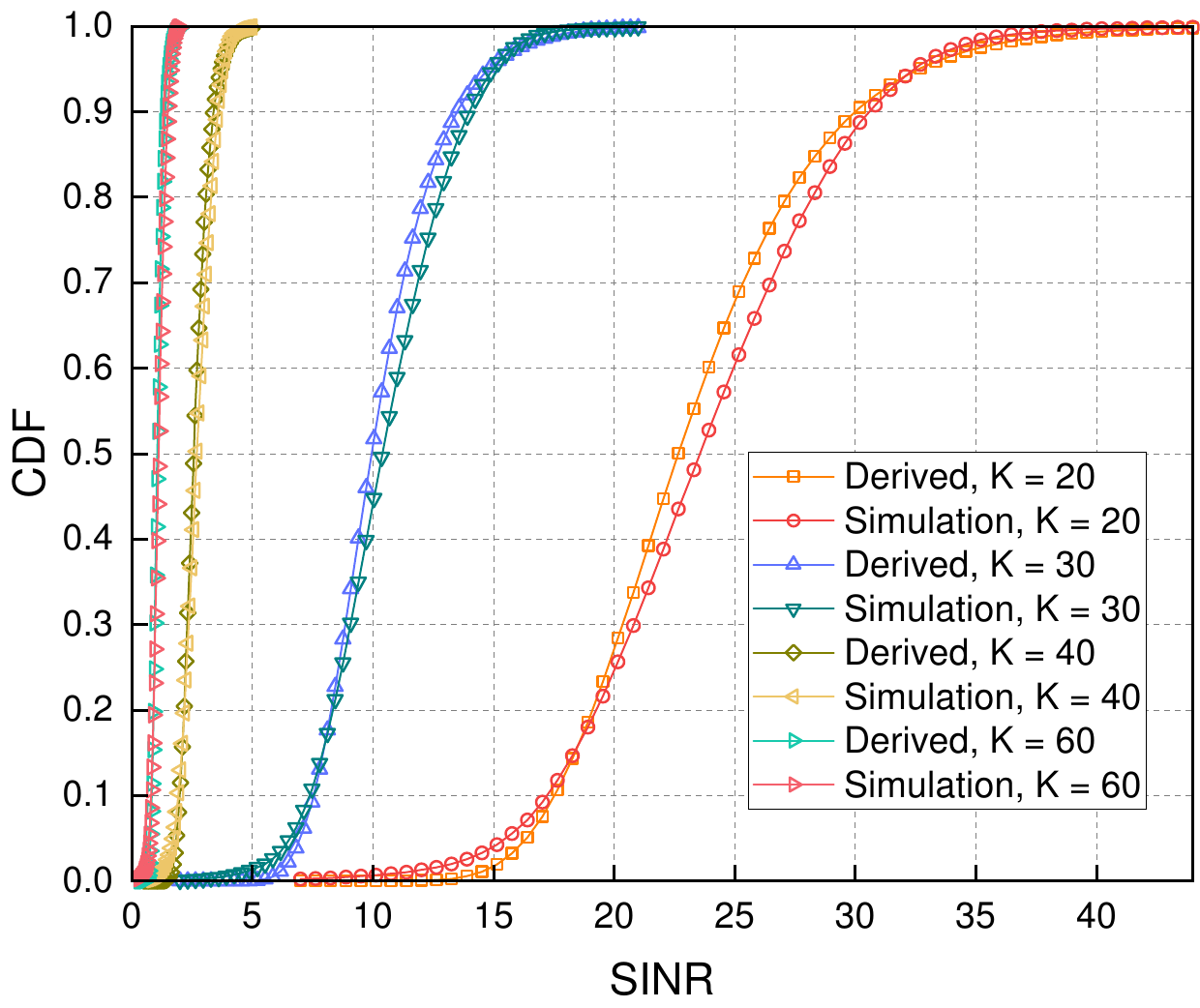}
    \caption{CDF of SINR under FZF precoding with different number of users. System parameters: $M$ = 120, $N$ = 11, $l_p$ = 10.}
	\label{CDFKFZF}
\end{figure}

\begin{figure}[ht]
	\centering
    \setlength{\abovecaptionskip}{-0.1cm}   
	\includegraphics[scale=0.35]{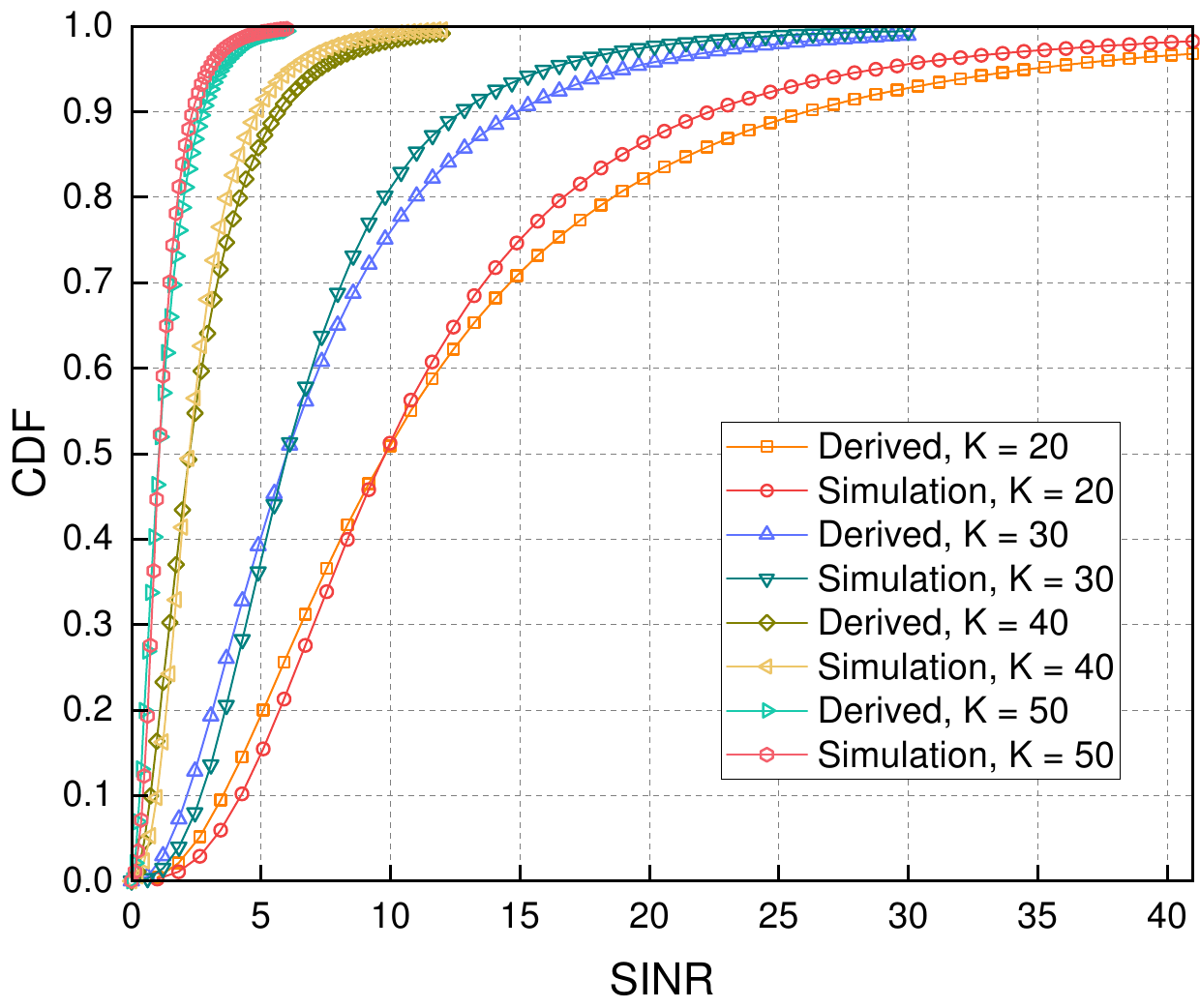}
    \caption{CDF of SINR under MRT precoding with different number of users. System parameters: $M$ = 120, $N$ = 2, $l_p$ = 10.}
	\label{CDFKMRT}
\end{figure}

\begin{figure}[ht]
	\centering
    \setlength{\abovecaptionskip}{-0.1cm}   
	\includegraphics[scale=0.35]{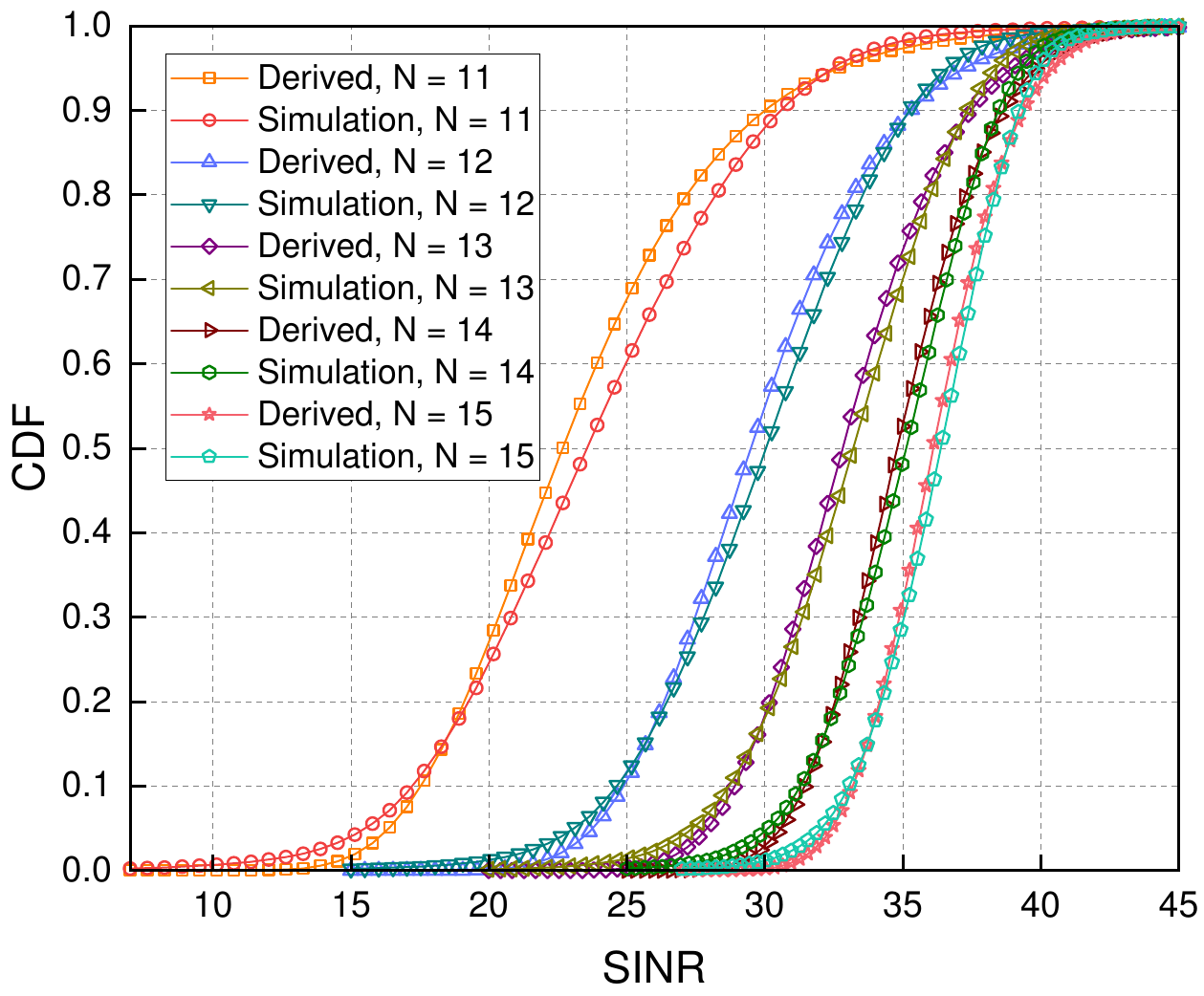}
    \caption{CDF of SINR under FZF precoding with different number of antennas each AP. System parameters: $M$ = 120, $K$ = 20, $l_p$ = 10.}
	\label{CDFNFZF}
\end{figure}

\begin{figure}[ht]
	\centering
    \setlength{\abovecaptionskip}{-0.1cm}   
	\includegraphics[scale=0.35]{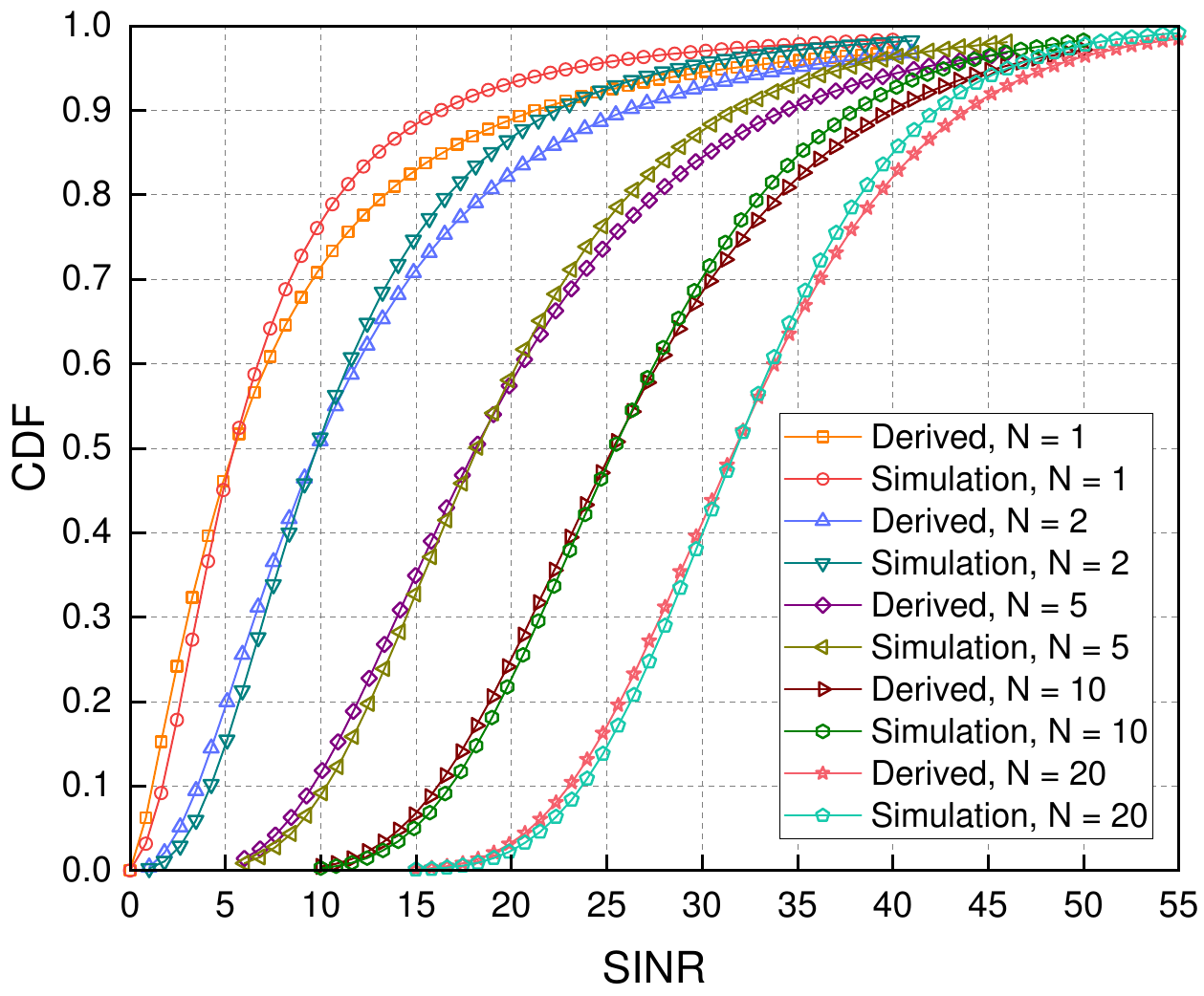}
    \caption{CDF of SINR under MRT precoding with different number of antennas each AP. System parameters: $M$ = 120, $K$ = 20, $l_p$ = 10.}
	\label{CDFNMRT}
\end{figure}

We compare the CDF of the SINR for varying numbers of users under FZF precoding, as illustrated in Fig. \ref{CDFKFZF}. It is evident that as the number of users increases, the derived results closely converge with the Monte Carlo simulation results, eventually reaching a consistent agreement. This convergence is attributed to our approximation of the IN based on the CLT. As the number of users grows, this approximation becomes increasingly accurate and approaches the true results. Furthermore, with an increasing number of users, the SINR correspondingly decreases. Despite the unchanged length of the pilot sequence, more users share the same pilot, resulting in significant pilot contamination and subsequently causing severe interference, ultimately leading to a lower SINR.

In Fig. \ref{CDFKMRT}, we present the CDF of the SINR under the deployment of MRT precoding, showcasing its variation with respect to the number of users. As the number of users increases, the derived distribution of SINR exhibits a growing consistency with the results obtained from Monte Carlo simulations. This increasing alignment can be attributed to the heightened applicability of the CLT as the number of users expands. Furthermore, similar to the observation in the case of FZF precoding, the SINR decreases with an increasing number of users. The reason for this degradation in SINR is the same. Additionally, under MRT precoding, users who do not employ the same pilot sequence also introduce interference with each other. Consequently, employing FZF precoding in CF mMIMO systems generally yields higher SINR, indicating superior performance compared to MRT precoding.

We turn our attention to comparing the impact of varying the number of antennas on each AP under FZF precoding on the CDF of the SINR in CF mMIMO systems, as depicted in Fig. \ref{CDFNFZF}. It is evident that as the number of antennas increases, the derived results align more closely with the results obtained from system analysis. This convergence can be attributed to the increased number of constituent signals in the DS and IN due to the growing number of antennas, thereby enhancing the accuracy of the approximation based on the CLT. Furthermore, the increase in the number of antennas leads to an increase in SINR. From (\ref{expression of Uk1 in FZF}), it can be observed that DS increases with the growing number of antennas. While, according to (\ref{expression of Uk2 in FZF}), the $U_{kk_1}^2$ component of IN also increases with the number of users, (\ref{second and fourth of psi downlink FZF}) indicates that $U_{kk_1}^3$ remains unchanged with an increasing number of antennas. Hence, an increase in the number of antennas results in an improvement in SINR.

Fig. \ref{CDFNMRT} illustrates the distribution of SINR under MRT precoding for different numbers of antennas. Similar to the scenario with FZF precoding, it is observed that as the number of antennas increases, the derived CDF of SINR closely aligns with the results obtained from Monte Carlo simulations. Furthermore, an increase in the number of antennas leads to an improvement in SINR. As MRT precoding only mitigates the influence of noise, while FZF precoding suppresses the effects of interference, FZF precoding generally achieves better performance.
\subsection{Achievable Rate and Outage Probability Analysis}

\begin{figure}[ht]
	\centering
    \setlength{\abovecaptionskip}{-0.1cm}   
	\includegraphics[scale=0.35]{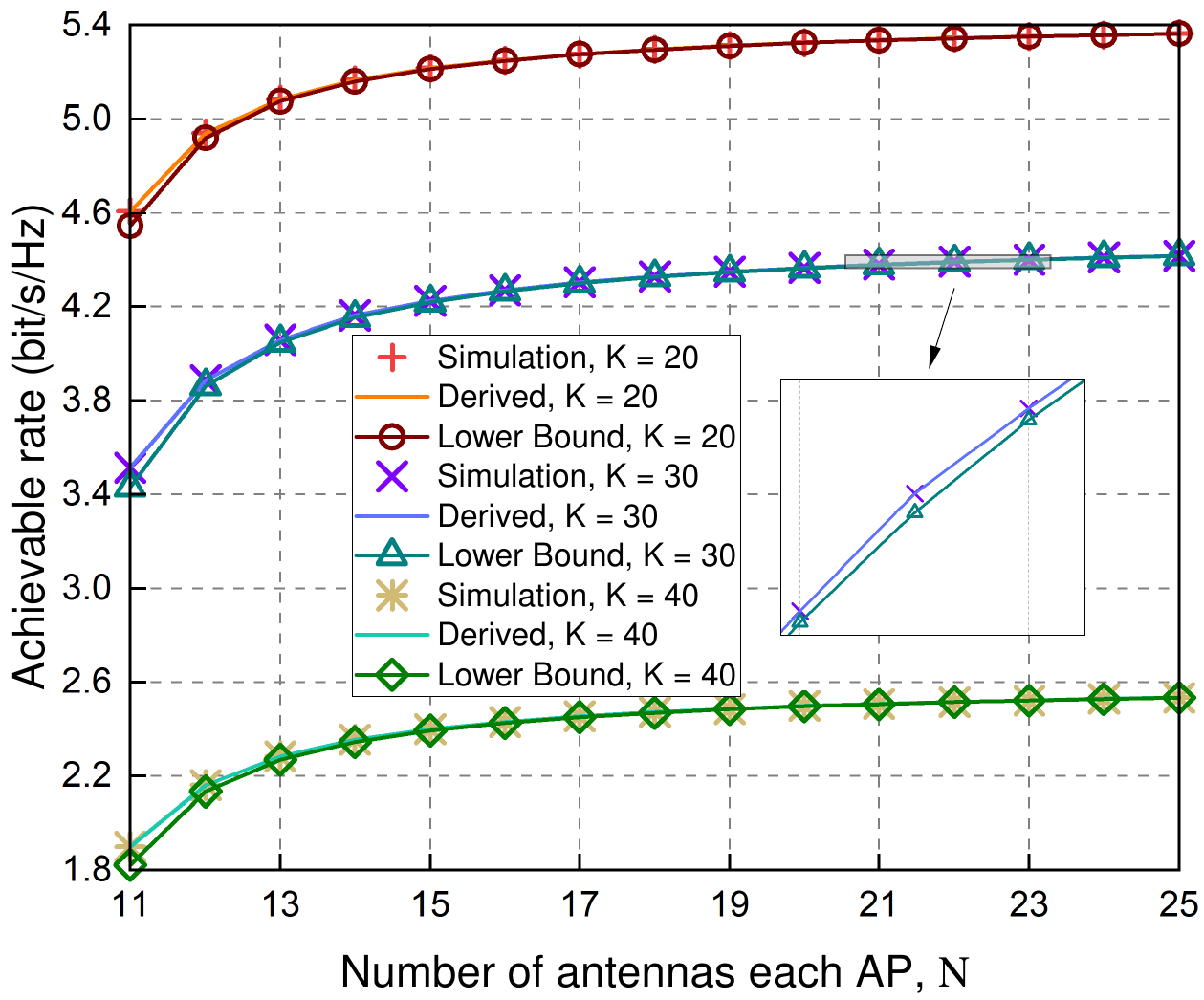}
    \caption{Achievable rate under FZF precoding with different number of antennas each AP. System parameters: $M$ = 120, $l_p$ = 10.}
	\label{RateFZF}
\end{figure}

\begin{figure}[ht]
	\centering
    \setlength{\abovecaptionskip}{-0.1cm}   
	\includegraphics[scale=0.35]{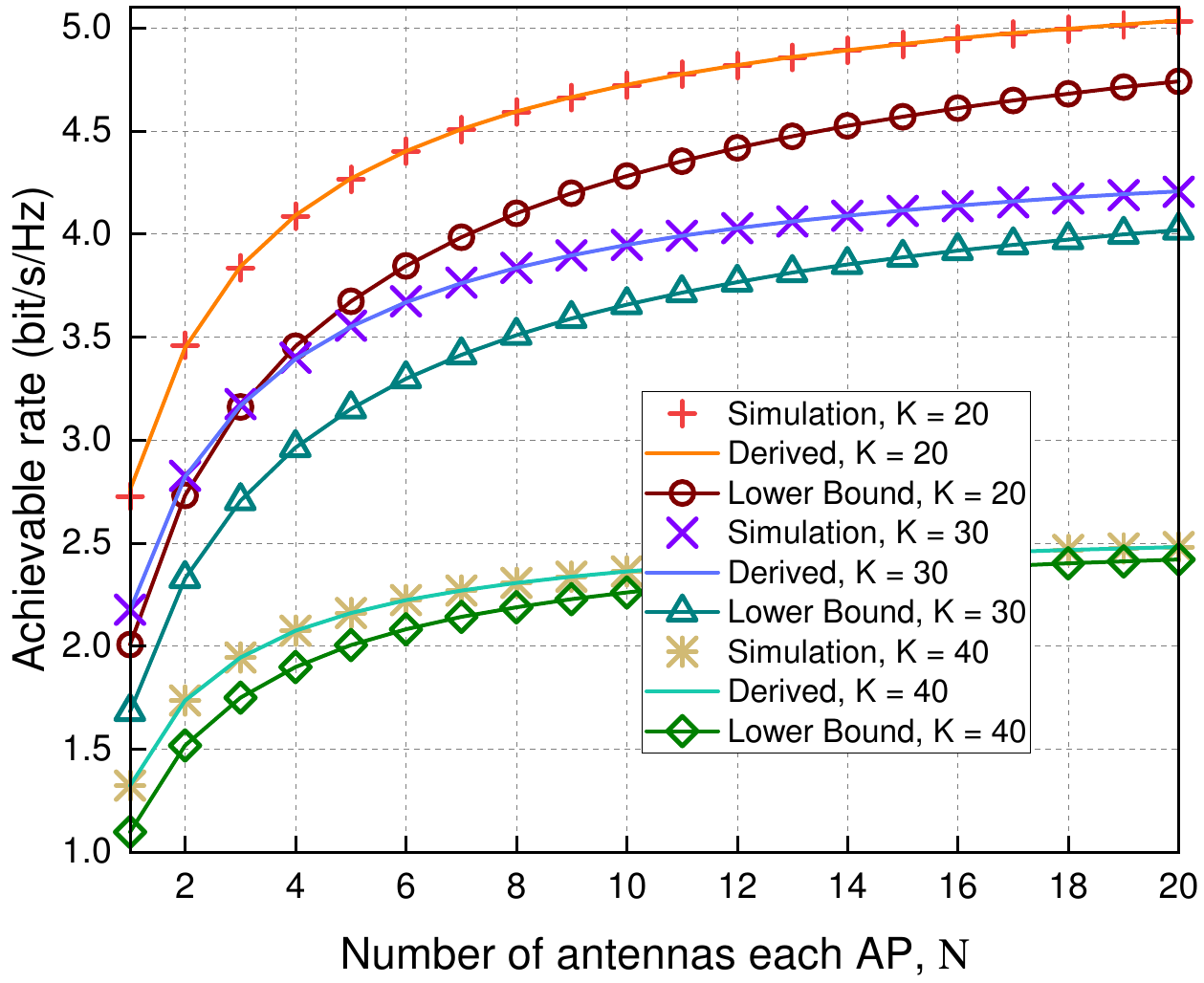}
    \caption{Achievable rate under MRT precoding with different number of antennas each AP. System parameters: $M$ = 120, $l_p$ = 10.}
	\label{RateMRT}
\end{figure}

We have investigated the impact of varying the number of antennas on each AP and the number of users on achievable rates under both FZF precoding and MRT precoding, as depicted in Fig. \ref{RateFZF} and Fig. \ref{RateMRT}, respectively. It can be observed that the results derived from our analysis align with the results obtained from Monte Carlo simulations. As the number of users decreases and the number of antennas increases, the achievable rates also increase. This is because the reduction in the number of users and the increase in the number of antennas result in higher SINR, thereby leading to an increase in achievable rates.
Besides, it can be observed that our derived results are higher than the lower bound with different scenarios. Under MRT precoding, our derived results are significantly higher than the lower bound. However, under FZF precoding, our derived results are only slightly higher than the lower bound.
This is because, with FZF precoding, both desired signals and interference can be directly computed, and the uncertainty in SINR mainly comes from channel estimation errors. Therefore, compared to MRT precoding, FZF precoding yields a more stable SINR, resulting in actual achievable rates that are closer to the lower bound.

\begin{figure}[ht]
	\centering
    \setlength{\abovecaptionskip}{-0.1cm}   
	\includegraphics[scale=0.35]{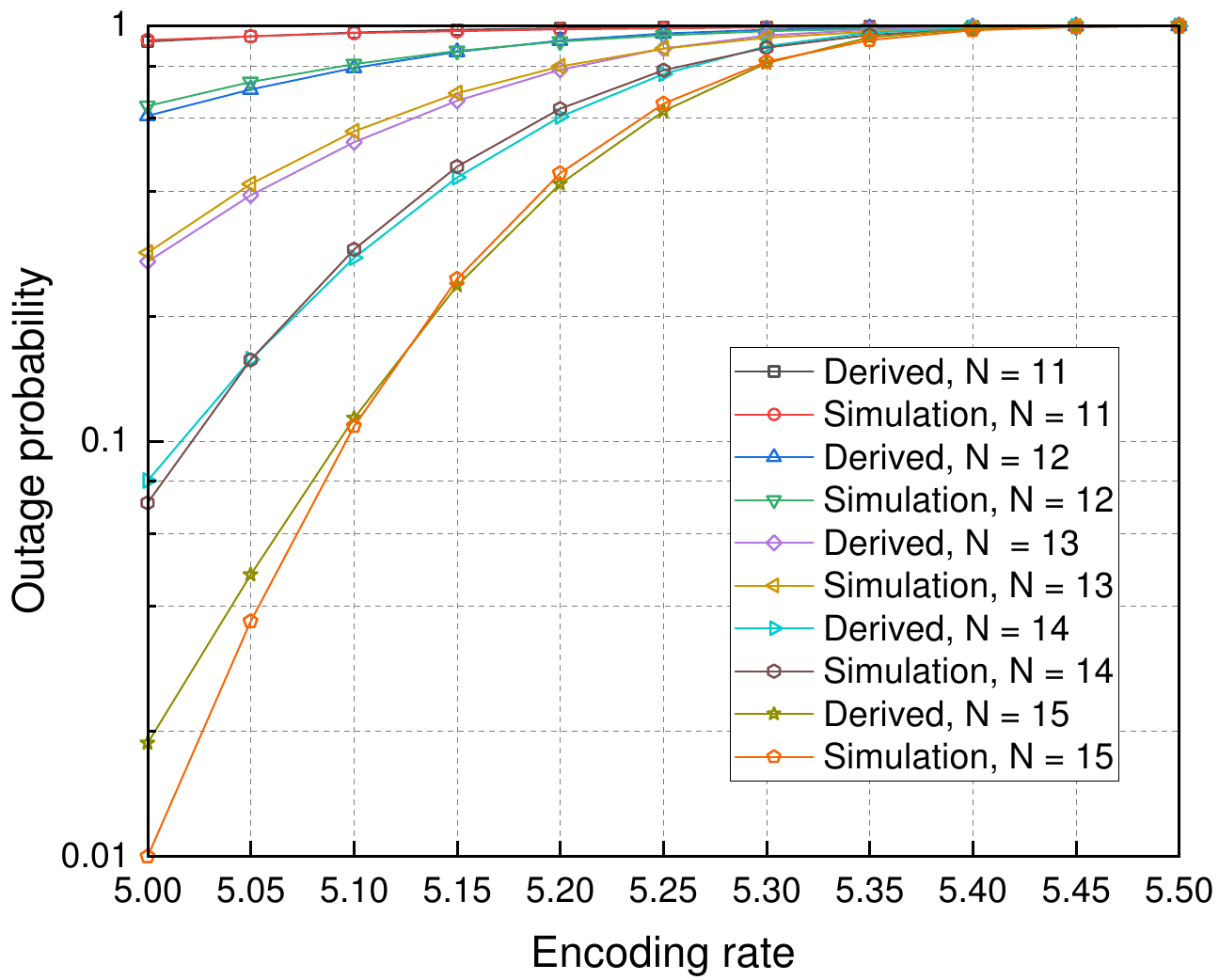}
    \caption{Outage probability under FZF precoding with different number of antennas each AP. System parameters: $M$ = 120, $l_p$ = 10.}
	\label{OutageFZF}
\end{figure}

\begin{figure}[ht]
	\centering
    \setlength{\abovecaptionskip}{-0.1cm}   
	\includegraphics[scale=0.35]{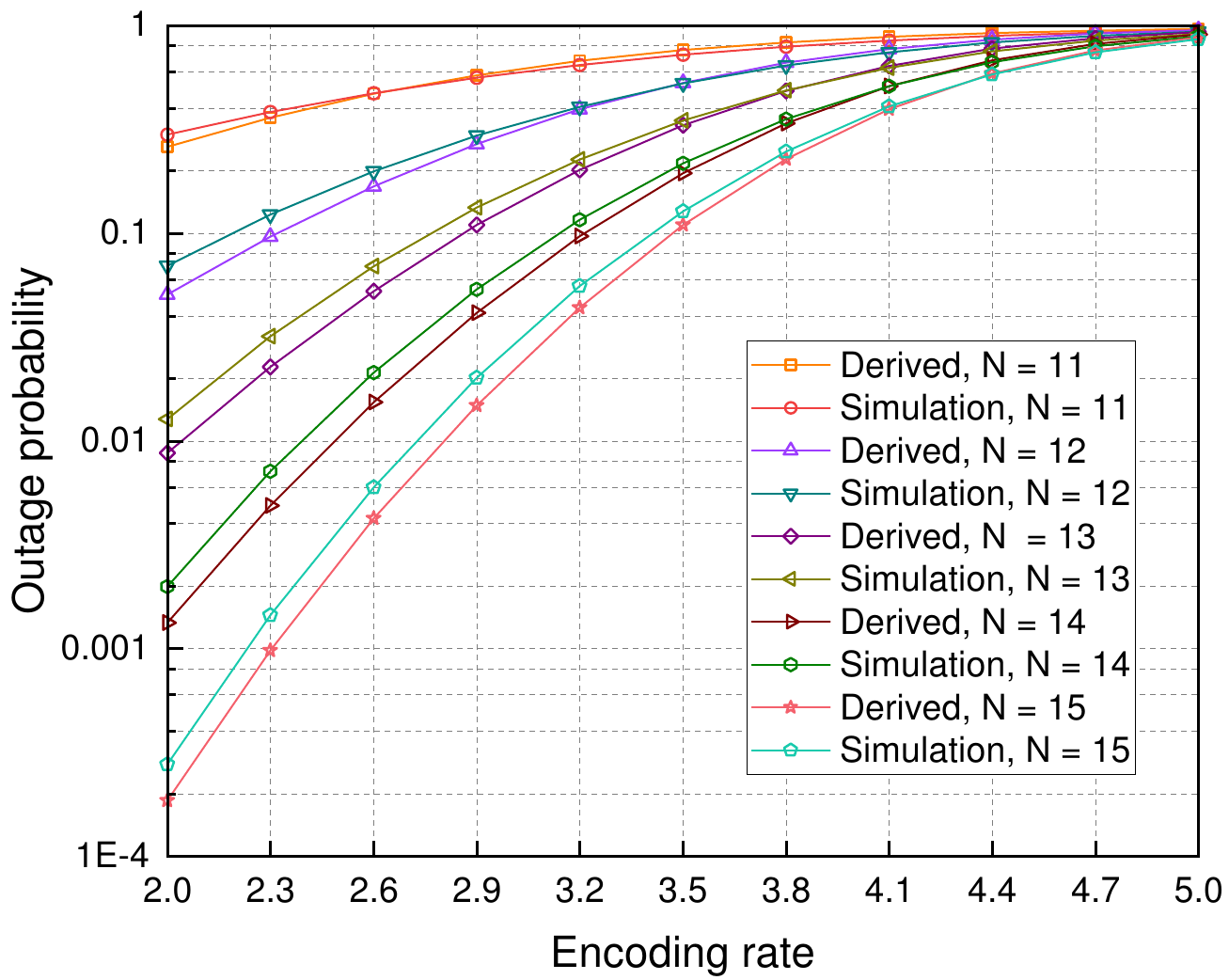}
    \caption{Outage probability under MRT precoding with different number of antennas each AP. System parameters: $M$ = 120, $l_p$ = 10.}
    \label{OutageMRT}
\end{figure}

We conducted separate investigations on the impact of different coding rates and the number of antennas on each AP for outage probability under FZF precoding and MRT precoding, as depicted in Fig. \ref{OutageFZF} and Fig. \ref{OutageMRT}, respectively. As anticipated, the outage probability increases continuously with higher coding rates and a reduced number of antennas on the APs. Moreover, it is noteworthy that our derived results closely align with the results obtained from Monte Carlo simulations across various scenarios.

\section{Conclusion}
In this paper, we provided the PDF and CDF of the SINR in CF mMIMO systems considering both MRT and FZF precoding schemes.
Moreover, we have performed a comprehensive performance analysis based on these results.
Specifically, we modeled a CF mMIMO system with pilot contamination and derived the expression for the SINR. Then, by considering the independence of the signals, we divided the SINR into two components, i.e., DS and IN.
Under MRT precoding, we derived the distributions of DS and IN using CLT and random matrix theory, enabling the analysis of the distribution of the SINR. Under FZF precoding, we directly computed the DS and derived the distribution of IN, which allowed us to obtain the PDF and CDF of the SINR.
Based on the aforementioned analysis, we derived expressions for the achievable rate and the outage probability in CF mMIMO systems under MRT and FZF precoding, respectively. Finally, simulation results demonstrated the accuracy of our derivations.

Currently, research on the statistical characteristics of SINR in the CF mMIMO system is still in its preliminary stage, and further investigation is required.
Firstly, when deploying other precoding schemes, such as local ZF or MMSE precoding, in CF mMIMO systems, the corresponding SINR distributions also need to be analyzed. Secondly, in a more realistic CF mMIMO system, factors such as backhaul constraints and hardware impairments must be considered to evaluate their impact on the SINR distribution and derive appropriate models. Lastly, it would be a meaningful endeavor to design new power allocation schemes, pilot allocation schemes, and other system design schemes based on the SE and EE derived from our current analysis of the statistical characteristics of SINR in CF mMIMO systems. These endeavors aim to better harness the potential of CF mMIMO systems.

\begin{appendices}

\section{Proof of Lemma \ref{lemma of U1 downlink MRT}}\label{appendix A}
  When MRT precoding is used for downlink data transmission, the precoding vector can be expressed as
  \begin{equation}\label{decoder vector of MRC}
    \mathbf{b}_{mk} = \frac{\bar{\mathbf{H}}_{m}\mathbf{e}_{i_k}}{\sqrt{\mathbb{E} \left \{ \left \| \bar{\mathbf{H}}_{m}\mathbf{e}_{i_k} \right \|^2 \right \}}   } = \frac{\hat{\mathbf{h}}_{mk} }{\sqrt{Nc_{mk}} }.
  \end{equation}

  Estimated channel $\hat{\mathbf{h}}_{mk}$ follows the distribution of $\mathcal{CN}(0,c_{mk}\mathbf{I}_N)$, then $\xi_{mk} =\sqrt{\frac{\eta_{mk}}{Nc_{mk} }}\hat{\mathbf{h}}_{mk}^H\hat{\mathbf{h}}_{mk}$ follows Gamma distribution with shape parameter $N$ and scale parameter $\sqrt{\frac{\eta_{mk}c_{mk}}{N}}$, i.e., $\xi_{mk} \sim Gamma(N,\sqrt{\frac{\eta_{mk}c_{mk}}{N} } )$.

  Then the first moment to fourth moment of $\xi_{mk}$ can be expressed as follows:
  \begin{subequations}\label{first to second moment of xi downlink MRT}
    \begin{align}
      &\mathbb{E}\left \{ \xi_{mk} \right \} = \sqrt{N\eta_{mk}c_{mk}}, \label{first to second moment of xi downlink MRTA}\\
      &\mathbb{E}\left \{ \xi_{mk}^2  \right \} = (N+1) \eta_{mk}c_{mk},\label{first to second moment of xi downlink MRTB}\\
      &\mathbb{E}\left \{ \xi_{mk}^3  \right \} = \frac{(N+1)(N+2)}{\sqrt{N} } \eta_{mk}^\frac{3}{2}  c_{mk}^\frac{3}{2}, \label{first to second moment of xi downlink MRTC}\\
      &\mathbb{E}\left \{ \xi_{mk}^4 \right \} = \frac{(N+1)(N+2)(N+3)}{N}  \eta_{mk}^2c_{mk}^2.\label{first to second moment of xi downlink MRTD}
    \end{align}
  \end{subequations}
  The first moment of the $ U_{k}^1 = \left |\sum_{m\in \mathcal{M}}\xi_{mk}\right |^2$ can be expressed as follows:
  \begin{equation}\label{first order of U1 downlink MRT}
    \begin{aligned}
        u_{U_{k}^1} &= \mathbb{E}\left \{ U_{k}^1 \right \} \\
                    &= \mathbb{E}\left \{ \left |\sum_{m\in \mathcal{M}}\xi_{mk}\right |^2 \right \} \overset{(a)}{=}  \mathbb{E}\left \{ \left (\sum_{m\in \mathcal{M}}\xi_{mk}\right )^2 \right \}\\
                    &\overset{(b)}{=}  \sum_{m\in \mathcal{M}} \mathbb{E}\left \{ \xi_{mk}^2  \right \}  + \sum_{m\in \mathcal{M}}\sum_{m_1 \ne m}\mathbb{E}\left \{ \xi_{mk} \right \}\mathbb{E}\left \{ \xi_{m_1k} \right \},
    \end{aligned}
  \end{equation}
  where $(a)$ follows the fact that $\xi_{mk}$ is real variable, and $(b)$ is obtained based on the independence between $\xi_{mk}$ and $\xi_{m'k}$ when $m \ne m'$.
  By inserting (\ref{first to second moment of xi downlink MRTA}) and (\ref{first to second moment of xi downlink MRTB}) into (\ref{first order of U1 downlink MRT}), we can obtain the first moment of the $U_k^1$.

  Similarly, the second moment of $U_k^1$ is given in (\ref{second order of U1 downlink MRT}).
  The second moment of $U_k^1$ is obtained by using (\ref{first to second moment of xi downlink MRTA})-(\ref{first to second moment of xi downlink MRTD}) and (\ref{second order of U1 downlink MRT}).
  \begin{figure*}
      \begin{equation}\label{second order of U1 downlink MRT}
        \begin{aligned}
           u_{U_k^1}^{(2)} &= \mathbb{E}\left \{  \left ( U_{k}^1  \right ) ^2 \right \} = \mathbb{E}\left \{ \left |\sum_{m\in \mathcal{M}}\xi_{mk}\right |^4 \right \}= \mathbb{E}\left \{ \left (\sum_{m\in \mathcal{M}}\xi_{mk}\right )^4 \right \} = \mathbb{E}\left \{ \left (\sum_{m\in \mathcal{M}}\xi_{mk}^2+\sum_{m\in \mathcal{M}}\sum_{m1 \ne m}\xi_{mk}\xi_{m1k}  \right )^2 \right \} \\
                       &=\sum_{m\in \mathcal{M}}\mathbb{E} \left \{\xi_{mk}^4 \right \} + \sum_{m\in \mathcal{M}}\sum_{ m_1 \ne m}\sum_{ m_2 \ne m,m_1} 6\mathbb{E} \left \{  \xi_{mk} ^2\right \} \mathbb{E} \left \{ \xi_{m_1k} \right \} \mathbb{E} \left \{ \xi_{m_2k}  \right \}  + \sum_{m\in \mathcal{M}}\sum_{m_1 \ne m}4\mathbb{E} \left \{ \xi_{mk}^3\right \} \mathbb{E} \left \{   \xi_{m_1k} \right \}\\
                       &\quad +\sum_{m\in \mathcal{M}}\sum_{m_1 \ne m}3\mathbb{E} \left \{\xi_{mk}^2 \right \} \mathbb{E} \left \{  \xi_{m_1k}^2 \right \} + \sum_{ m\in \mathcal{M}}\sum_{ m_1 \ne m}\sum_{ m_2 \ne m,m_1}\sum_{m_3 \ne m,m_1,m_2}\mathbb{E} \left \{ \xi_{mk}\right \} \mathbb{E} \left \{ \xi_{m_1k}  \right \}  \mathbb{E} \left \{ \xi_{m_2k}  \right \}  \mathbb{E} \left \{ \xi_{m_3k}  \right \}.\\
        \end{aligned}
      \end{equation}
	  {\noindent} \rule[-10pt]{18cm}{0.05em}
  \end{figure*}

\section{Proof of Lemma \ref{lemma of IN downlink MRT}}\label{appendix B}
  According to (\ref{parallel channel}), users using the same pilot sequence have the parallel estimated channels. Thus the first and second moments of $U_{kk_1}^2$ can be obtained similarly to $U_{k}^1$ when user $k$ and user $k_1$ use the same pilot sequence. The first to fourth moment of $\xi_{mkk_1}$ can be calculated as follows:
  \begin{subequations}\label{first to fourth moment of xi2 same}
    \begin{align}
      &\mathbb{E}\left \{ \xi_{mkk_1} \right \} = \sqrt{N\eta_{mk_1}c_{mk}}, \label{first to fourth moment of xi2 sameA}\\
      &\mathbb{E}\left \{ \xi_{mkk_1}^2  \right \} = (N+1) \eta_{mk_1}c_{mk}, \label{first to fourth moment of xi2 sameB}\\
      &\mathbb{E}\left \{ \xi_{mkk_1}^3  \right \} = \frac{(N+1)(N+2)}{\sqrt{N} } \eta_{mk_1}^\frac{3}{2}  c_{mk}^\frac{3}{2}, \label{first to fourth moment of xi2 sameC}\\
      &\mathbb{E}\left \{ \xi_{mkk_1}^4 \right \} = \frac{(N+1)(N+2)(N+3)}{N} \eta_{mk_1}^2c_{mk}^2. \label{first to fourth moment of xi2 sameD}
    \end{align}
  \end{subequations}
  Then the first moment of $U_{kk_1}^2$ can be expressed as follows:
  \begin{equation}\label{first moment of Uk2 downlink MRT same}
    \begin{aligned}
      &\mathbb{E}\left \{ U_{kk_1}^2 \right \} =  \\
      &\sum_{m \in \mathcal{M} }\mathbb{E}\left \{ \xi_{mkk_1}^2 \right \}    + \sum_{m \in \mathcal{M} }\sum_{m_1 \ne m} \mathbb{E}\left \{ \xi_{mkk_1} \right \} \mathbb{E}\left \{ \xi_{m_1kk_1}  \right \}
    \end{aligned}
  \end{equation}
  The second moment of $U_{kk_1}^2$ is given in (\ref{second moment of Uk2 downlink MRT 2 same}).
  The first and second moments of $U_{kk_1}^2$ can be obtained by using (\ref{first to fourth moment of xi2 same}), (\ref{first moment of Uk2 downlink MRT same}) and (\ref{second moment of Uk2 downlink MRT 2 same}).
  \begin{figure*}
      \begin{equation}\label{second moment of Uk2 downlink MRT 2 same}
        \begin{aligned}
           \mathbb{E}\left \{ \left ( U_{kk_1}^2 \right )^2  \right \}  &=\sum_{m\in \mathcal{M}}\mathbb{E} \left \{\xi_{mkk_1}^4 \right \} + \sum_{m\in \mathcal{M}}\sum_{ m_1 \ne m}\sum_{ m_2 \ne m,m_1} 6\mathbb{E} \left \{  \xi_{mkk_1} ^2\right \} \mathbb{E} \left \{ \xi_{m_1kk_1} \right \} \mathbb{E} \left \{ \xi_{m_2kk_1}  \right \}  \\
                       &\quad  + \sum_{m\in \mathcal{M}}\sum_{m_1 \ne m}4\mathbb{E} \left \{ \xi_{mkk_1}^3\right \} \mathbb{E} \left \{   \xi_{m_1kk_1} \right \} +\sum_{m\in \mathcal{M}}\sum_{m_1 \ne m}3\mathbb{E} \left \{\xi_{mkk_1}^2 \right \} \mathbb{E} \left \{  \xi_{m_1kk_1}^2 \right \} \\
                       &\quad + \sum_{ m\in \mathcal{M}}\sum_{ m_1 \ne m}\sum_{ m_2 \ne m,m_1}\sum_{m_3 \ne m,m_1,m_2}\mathbb{E} \left \{ \xi_{mkk_1}\right \} \mathbb{E} \left \{ \xi_{m_1kk_1}  \right \}  \mathbb{E} \left \{ \xi_{m_2kk_1}  \right \}  \mathbb{E} \left \{ \xi_{m_3kk_1}  \right \}.\\
        \end{aligned}
      \end{equation}
	  {\noindent} \rule[-10pt]{18cm}{0.05em}
  \end{figure*}

  When user $k$ and user $k_1$ use different pilot sequences, i.e., $k_1 \notin \mathcal{P}_k $, the second and fourth moment of $\xi_{mkk_1}$ can be expressed as follows:
  \begin{subequations}\label{first to fourth moment of xi2 different}
    \begin{align}
      &\mathbb{E}\left \{ \xi_{mkk_1}^*\xi_{mkk_1}  \right \}  = \eta_{mk_1}c_{mk}, \label{first to fourth moment of xi2 differentA}\\
      &\mathbb{E}\left \{ \left ( \xi_{mkk_1}^*\xi_{mkk_1} \right )^2  \right \}  = \frac{2\left ( N+1 \right ) }{N}  \eta_{mk_1}^2c_{mk}^2, \label{first to fourth moment of xi2 differentB}
    \end{align}
  \end{subequations}
  The first moment of $U_{kk_1}^2$ when $k_1 \notin \mathcal{P}_k $ can be expressed as follows:
  \begin{equation}\label{first1 moment of Uk2 downlink MRT}
    \begin{aligned}
      &\mathbb{E}\left \{ U_{kk_1}^2 \right \} = \mathbb{E}\left \{ \left ( \sum_{m \in \mathcal{M} } \xi_{mkk_1}^* \right )\left ( \sum_{m \in \mathcal{M} } \xi_{mkk_1} \right )   \right \}\\
      &=  \sum_{m \in \mathcal{M} }\mathbb{E}\left \{ \xi_{mkk_1}^*\xi_{mkk_1} \right \} + \sum_{m \in \mathcal{M} }\sum_{m_1 \ne m} \mathbb{E}\left \{ \xi_{mkk_1}^*  \xi_{m_1kk_1}  \right \}\\
      &\overset{(a)}{=}   \sum_{m \in \mathcal{M} }\mathbb{E}\left \{ \xi_{mkk_1}^*\xi_{mkk_1} \right \}
    \end{aligned}
  \end{equation}
  where $(a)$ is obtained based on the fact that the exception of $\xi_{mkk_1}$ is zero when user $k$ and user $k_1$ use different pilot sequences and the independence between $\xi_{mkk_1}$ and $\xi_{m_1kk_1}$ when $m \ne m_1$.
  Then the second moment of $U_{kk_1}^2$ when user $k$ and user $k_1$ use different pilot sequences is given in (\ref{second moment of Uk2 downlink MRT downlink different}).
  \begin{figure*}
      \begin{equation}\label{second moment of Uk2 downlink MRT downlink different}
        \begin{aligned}
          \mathbb{E}\left \{ \left ( U_{kk_1}^2 \right )^2  \right \} &= \mathbb{E}\left \{ \left | \sum_{m \in \mathcal{M} } \xi_{mkk_1} \right |^4  \right \} =\mathbb{E}\left \{ \left ( \sum_{m \in \mathcal{M} } \xi_{mkk_1}^*\xi_{mkk_1} + \sum_{m \in \mathcal{M} }\sum_{m_1 \ne m} \xi_{mkk_1}^*\xi_{m_1kk_1} \right )^2  \right \} \\
          & = \sum_{m \in \mathcal{M}}\mathbb{E}\left \{ \left ( \xi_{mkk_1}^*\xi_{mkk_1} \right )^2 \right \} + 2\sum_{m \in \mathcal{M} }\sum_{m_1 \ne m}\mathbb{E}\left \{ \xi_{mkk_1}^*\xi_{mkk_1}\right \} \mathbb{E}\left \{ \xi_{m_1kk_1}^*\xi_{m_1kk_1} \right \}.
        \end{aligned}
      \end{equation}
	  {\noindent} \rule[-10pt]{18cm}{0.05em}
  \end{figure*}

  By using (\ref{first to fourth moment of xi2 different}), (\ref{first1 moment of Uk2 downlink MRT}) and (\ref{second moment of Uk2 downlink MRT downlink different}), we can obtain the first and second moment of $U_{kk_1}^2$ when $k_1 \notin \mathcal{P}_k$.

  We calculated the first and second moments of $U_{kk_1}^3$ using a similar way with $U_{kk_1}^2$. The second and fourth moment of $\psi_{mkk_1}$ can be calculated as follows:
  \begin{equation}\label{second and fourth of psi downlink MRT}
    \begin{aligned}
      &\mathbb{E}\left \{ \psi_{mkk_1}^*\psi_{mkk_1} \right \}  = \eta_{mk_1} \left ( \beta_{mk}-c_{mk} \right ),\\
      &\mathbb{E}\left \{ \left ( \psi_{mkk_1}^*\psi_{mkk_1} \right )^2  \right \}  = \frac{2\left ( N+1 \right )}{N} \eta_{mk_1}^2\left ( \beta_{mk}-c_{mk} \right )^2 ,\\
    \end{aligned}
  \end{equation}
  Then the first and second moment of $U_{kk_1}^3$ can be expressed as follow:
  \begin{equation}\label{first and second moment of Uk3}
    \begin{aligned}
      &\mathbb{E}\left \{ U_{kk_1}^3 \right \} = \sum_{m\in \mathcal{M}}\mathbb{E}\left \{ \psi_{mkk_1}^* \psi_{mkk_1}\right \},\\
      &\mathbb{E}\left \{ \left ( U_{kk_1}^3 \right )^2  \right \} = \sum_{m \in \mathcal{M}}\mathbb{E}\left \{ \left ( \psi_{mkk_1}^*\psi_{mkk_1} \right )^2 \right \} \\
      &+ 2\sum_{m \in \mathcal{M} }\sum_{m_1 \ne m}\mathbb{E}\left \{ \psi_{mkk_1}^*\psi_{mkk_1}\right \} \mathbb{E}\left \{ \psi_{m_1kk_1}^*\psi_{m_1kk_1} \right \}.
    \end{aligned}
  \end{equation}
  The first and second moment of $U_{kk_1}^3$ can be obtained by using (\ref{second and fourth of psi downlink MRT}) and (\ref{first and second moment of Uk3}).
  The first order of $\text{IN}_k^{\text{FZF}}$ can be expressed as follows:
  \begin{equation}\label{fisrst order of IN downlink MRT}
    u_{\text{IN}_{k}^{\text{FZF}}} = \rho_d\sum_{k_1\ne k} \mathbb{E}\left \{ U_{kk_1}^2 \right \} +  \rho_d\sum_{k_1\in \mathcal{K} }\mathbb{E}\left \{ U_{kk_1}^3 \right \}  + 1.
  \end{equation}
  The second moments of $\text{IN}_{k}^{\text{FZF}}$ can be expressed as

  \begin{figure*}
      \begin{equation}\label{second order of IN downlink MRT}
        \begin{aligned}
          &u_{\text{IN}_{k}^{\text{FZF}}}^{(2)} = \mathbb{E}\left \{ \left (\rho_d\sum_{k_1\ne k}U_{kk_1}^2+\rho_d\sum_{k_1 \in \mathcal{K} }U_{kk_1}^3+z_k\right )^2  \right \} \\
                       &= \rho_d^2\left ( \sum_{k_1\ne k}\mathbb{E}\left \{ \left ( U_{kk_1}^2 \right )^2  \right \} + \sum_{k_1\ne k}\sum_{k_2\ne k,k_1}\mathbb{E}\left \{ U_{kk_1}^2 \right \} \mathbb{E}\left \{ U_{kk_2}^2 \right \} + \sum_{k_1\in \mathcal{K} }\mathbb{E}\left \{ \left ( U_{kk_1}^3 \right )^2 \right \} + 2\right . \\
                       & \left .  + \sum_{k_1\in \mathcal{K} }\sum_{k_2\ne k_1}\mathbb{E}\left \{ U_{kk_1}^3 \right \} \mathbb{E}\left \{ U_{kk_2}^3 \right \}  +  2\left ( \sum_{k_1\ne k}\mathbb{E}\left \{ U_{kk_1}^2 \right \} \right ) \left ( \sum_{k_1\in \mathcal{K} }\mathbb{E}\left \{ U_{kk_1}^3 \right \}\right ) \right ) +  2\rho_d\left (\sum_{k_1\ne k}\mathbb{E}\left \{ U_{kk_1}^2 \right \}+ \sum_{k_1\in \mathcal{K} }\mathbb{E}\left \{ U_{kk_1}^3 \right \} \right ).
        \end{aligned}
      \end{equation}
	  {\noindent} \rule[-10pt]{18cm}{0.05em}
  \end{figure*}
  Then the first and second moments of $\text{IN}_{k}^{\text{FZF}}$ can be obtained.

\section{Proof of Theorem \ref{PDF of gamma when MRT downlink}}\label{appendix C}
  In the CF mMIMO system with MRT precoding, the distribution of $\text{DS}_k^{\text{MRT}}$ can be approximated as a Gamma distribution with shape parameter $j_{k1}$ and scale parameter $\rho_d\chi_{k1}$ for user $k$, and $\text{IN}_k^{\text{MRT}}$ can be seemed as Gamma distribution with shape parameter $j_{k1}$ and scale parameter $\chi_{k2}$.
  Since the $\rho_d U_{k1}$ and $\text{IN}_k^{\text{MRT}}$ are independent with each other. The CDF of $\gamma_k$ can be calculated as follows:
  \begin{equation}\label{proof of ration 1}
    \begin{aligned}
      P\left \{ \gamma_k \le x \right \} &= P \left \{ \frac{\text{DS}_k^{\text{MRT}}}{\text{IN}_k^{\text{MRT}}}\le x\right \} =  P \left \{ \text{DS}_k^{\text{MRT}}\le \text{IN}_k^{\text{MRT}}x\right \}\\
                                         &= \int_{0}^{\infty}F_{\text{DS}_k^{\text{MRT}}}(xx_1)f_{\text{IN}_k^{\text{MRT}}}(x_1)dx_1.
    \end{aligned}
  \end{equation}
  where $F_{\text{DS}_k^{\text{MRT}}}(\cdot )$ denotes the CDF of $\text{DS}_k^{\text{MRT}}$ and $f_{\text{IN}_k^{\text{MRT}}}(\cdot)$ represnets the PDF of $\text{IN}_k^{\text{MRT}}$.
  Then the PDF of $\gamma_k$ can be calculated based on (\ref{proof of ration 1}), which is given in (\ref{proof of ration 2})

  Based on the PDF of $\gamma_k$, the CDF of $\gamma_k$ can be calculated using $F_{\gamma_k}^{\text{MRT}}(x) = \int_{0}^{x} f_{\gamma_k}^{\text{MRT}}(x_1) dx_1 $ directly.

  \begin{figure*}
    \begin{equation}\label{proof of ration 2}
      \begin{aligned}
         f_{\gamma_k}^{\text{MRT}}(x) &= \frac{\mathrm{d} P\left \{ \gamma_k \le x \right \}}{\mathrm{d} x} = \int_{0}^{\infty}x_1f_{\text{DS}_k^{\text{MRT}} }(xx_1)f_{\text{IN}_k^{\text{MRT}}}(x_1)dx_1. \\
                      &=  \int_{0}^{\infty} x_1 \frac{1}{\Gamma(j_{k1})\left(\rho_d\chi_{k1} \right)^{j_{k1}}} \left ( xx_1 \right) ^{\left ( j_{k1}-1 \right ) } e^{-\frac{xx_1}{\chi_{k1}}} \frac{1}{\Gamma(j_{k2})\left(\chi_{k2} \right)^{j_{k2}}}\left ( x_1 \right ) ^{\left ( j_{k2}-1 \right ) }e^{-\frac{x_1}{\chi_{k2}}}dx_1 \\
                      &= \frac{\Gamma(j_{k1}+j_{k2})}{\Gamma(j_{k1})\Gamma(j_{k2})\left ( \rho_d\chi_{k1} \right ) ^{j_{k1}}\chi_{k2}^{j_{k2}}}x^{j_{k1}-1}(\frac{1}{\chi_{k2}}+\frac{x}{\rho_d\chi_{k1}}) ^{-j_{k1}-j_{k2}}.
     \end{aligned}
    \end{equation}
	{\noindent} \rule[-10pt]{18cm}{0.05em}
  \end{figure*}

\end{appendices}

\bibliographystyle{IEEEtran}
\bibliography{IEEEabrv,On_the_Distribution_of_SINR_for_Cell_Free_Massive_MIMO_Systems}

\end{document}